\documentclass[lettersize,journal]{IEEEtran}
\usepackage{amsmath,amsfonts}
\usepackage{ragged2e}
\usepackage{siunitx}
\usepackage{caption}
\usepackage{amsmath}
\usepackage{amssymb}
\usepackage{array}
\usepackage{textcomp}
\usepackage{stfloats}
\usepackage{url}
\usepackage{adjustbox}
\usepackage{verbatim}
\usepackage{graphicx}
\usepackage{cite}
\usepackage{tikz}
\usepackage[flushleft]{threeparttable}
\usepackage{amsthm}
\usepackage{subcaption}
\usepackage[normalem]{ulem}
\usepackage{etoolbox}
\usepackage{amssymb}
\newtheorem{remark}{Remark}

\usepackage{mathtools}

\usepackage{breqn}
\usepackage[multiple]{footmisc}

\definecolor{ForestGreen}{RGB}{34,139,34}
\usepackage{bbm}
\newtheorem{theorem}{Theorem}
\newtheorem{lemma}{Lemma}

\DeclareMathOperator*{\argmax}{arg\,max}

\allowdisplaybreaks

\usepackage{algorithm}
\usepackage{algpseudocode}
\usepackage{setspace}
\usepackage[font=footnotesize,labelfont=bf,skip=5pt]{caption}
\usepackage{subcaption}
\usepackage{enumitem}
\usepackage{bigints}
\usepackage[colorlinks]{hyperref}
\hypersetup{
	urlcolor     = magenta, 
	linkcolor    = blue, 
	filecolor	    = cyan, 
	citecolor   = {blue}, 
}
\newcommand\encircle[1]{%
	\tikz[baseline=(X.base)] 
	\node (X) [draw, shape=circle, inner sep=0] {\strut #1};}
\usepackage{scalerel}[2016-12-29]

\usepackage{wrapfig}
\usepackage[outdir=./]{epstopdf}
\begin{document}
\title{\vspace{-0.7cm}Performance Analysis of Multi-IRS Aided Multiple Operator Systems at mmWave Frequencies}
	\author{Souradeep Ghosh, L. Yashvanth,~\IEEEmembership{Student Member,~IEEE} and Chandra R. Murthy,~\IEEEmembership{Fellow,~IEEE}
		\thanks{S. Ghosh is with Qualcomm Inc., India. He was with the Dept. of ECE, Indian Institute of Science (IISc), Bengaluru, during the development of this work. L. Yashvanth and C. R. Murthy are with the Dept. of ECE, IISc, Bengaluru, India, 560012. (e-mails: souradeep2210@gmail.com, yashvanthl@iisc.ac.in, cmurthy@iisc.ac.in).}
		\thanks{The work of L. Yashvanth and Chandra R. Murthy was financially supported by the Prime Minister's Research Fellowship, Govt. of India and the Qualcomm 6G UR India Grant, respectively.}
		\thanks{S. Ghosh and L. Yashvanth contributed equally to this work.}
		\vspace{-0.58cm}
	}
        \maketitle
	\begin{abstract}
         Intelligent reflecting surfaces (IRSs) are envisioned to enhance the performance of mmWave wireless systems. In practice, multiple mobile operators (MO) coexist in an area and provide simultaneous and independent services to user-equipments (UEs) on different frequency bands. Then, if each MO deploys an IRS to enhance its performance, the IRSs also alter the channels of UEs of other MOs. In this context, this paper addresses the following questions: can an MO still continue to control its IRS independently of other MOs and IRSs? Is joint optimization of IRSs deployed by different MOs and inter-MO cooperation needed? To that end, by considering the mmWave bands, we first derive the ergodic sum spectral efficiency (SE) in a $2$-MO system for the following schemes: $1$) joint optimization of an overall phase angle of the IRSs with MO cooperation, $2$) MO cooperation via time-sharing, and $3$) no cooperation between the MOs. We find that even with no cooperation between the MOs, the performance of a given MO is not degraded by the presence of an out-of-band (OOB) MO deploying and independently controlling its own IRS. 
         On the other hand, the SE gain obtained at a given MO using joint optimization and cooperation over the no-cooperation scheme decreases inversely with the number of elements in the IRS deployed by the other MO. We generalize our results to a multiple MO setup and show that the gain in the sum-SE over the no-cooperation case increases at least linearly with the number of OOB MOs. Finally, we numerically verify our findings and conclude that every MO can independently operate and tune its IRS; cooperation via optimizing an overall phase only brings marginal benefits in practice.                  \end{abstract}
	\vspace{-0.3cm}
	\section{Introduction}
        \vspace{-0.1cm}
       Millimeter-wave (mmWave) frequency bands have been incorporated into current wireless standards to enable high data rates by leveraging the availability of large bandwidths~\cite{Dahlman_CSMag_2017}. However, a concern with the use of mmWave bands is the high propagation loss, which limits cellular coverage. Intelligent reflecting surfaces (IRSs) have recently been introduced to tackle this issue by providing virtual line-of-sight paths~\cite{RuiZhang_IRSSurvey_TCOM_2021}. Further, in real-life scenarios, multiple \textcolor{black}{mobile operators} (MOs) using different and non-overlapping frequency bands coexist in a geographical area and provide independent services to different \textcolor{black}{user equipments} (UEs) that are subscribed to them. In such a scenario, since an IRS is a passive device and does not contain a bandpass filter, it reflects the signals of every MO in the system. Thus, it remains unclear whether MOs can independently optimize their IRSs for their UEs or if cooperation among MOs is required due to the presence of IRSs. This paper addresses these issues and offers insights into multiple MO systems aided by IRSs in the mmWave bands.
       \vspace{-0.2cm}
 \subsection{Related Work \& Motivation}
 \textcolor{black}{The IRS literature has seen significant growth in recent years. For instance,\cite{RuiZhang_TWC_2019} and~\cite{Wang_TVT_2020} investigate the joint design of active and passive beamforming for sub-6 GHz and mmWave systems, respectively. In~\cite{Zheng_WCL_2020}, the authors propose channel estimation and IRS phase optimization techniques tailored to orthogonal frequency division multiplexing (OFDM)-based systems. The work in~\cite{Schober_TWC_2021} explores multiple access schemes in IRS-assisted networks, while~\cite{IRSvsRelay_WCL_2020} demonstrates the performance gains of IRSs over conventional relays. Additionally,~\cite{Rui_Security_WCL_2019} highlights the potential of IRSs in enhancing physical-layer security. The study in~\cite{Cunhua_TWC_2020} extends IRS deployment to multi-cell environments, and~\cite{WeiYu_JSAC_2021} introduces machine learning-based techniques for optimizing IRS reflection coefficients. A comprehensive review of the applications and benefits of IRSs in mmWave systems can be found in~\cite{RuiZhang_IRSSurvey_TCOM_2021,Liu_CST_2021_IRS_survey_2,Gong_CST_2020_IRS_survey_3}.}

In~\cite{Liu_ICCC_2022} and~\cite{Yang_ICC_2020}, hybrid beamforming architectures were proposed for IRS-aided mmWave systems using instantaneous and statistical channel state information (CSI), respectively. Following this, \cite{Zhao_TWC_2022_Jo_Pow_alloc} and~\cite{Song_TVT_ICI_Suppressing} solved for optimal power control coefficients/UE associations and IRS configurations, respectively, to maximize the sum-rate of UEs in multiple-IRS setups. 
In~\cite{Wang_TWC_2022_IRS_BA}, a beam training problem for IRSs exploiting the channel sparsity was solved, and in~\cite{Lin_TCOM_2022_Ch_Est_IRS} and~\cite{Yashvanth_GLOBECOMM_2022}, novel CSI estimation techniques were proposed for centralized and distributed IRS setups, respectively. However, all these works implicitly assume the presence of only one MO that deploys and controls the IRS. \textcolor{black}{The problem becomes more challenging when we consider more than one MO in the system. To explain, MOs are typically allotted non-overlapping frequency bands centered at nearby carrier frequencies to provide service to the UEs subscribed to them. So, the MO that deploys and controls the IRS (called the in-band MO) tunes the IRS phase configuration in the frequency band allotted to the MO to best serve its own UEs. However, the IRSs are passive, i.e., they do not contain any active signal processing/RF circuitry such as a band-pass filter to selectively reflect signals whose frequency content lies only within the band allotted to the in-band operator. As a result, any other MO providing service in the same geographical area in a nearby frequency allocation will naturally have its signals reflected off the IRS with an arbitrary phase shift. For example, the n$257$ band in $5$G new radio (NR) operates in the mmWave frequencies and spans $26.5$-$29.5$ GHz, i.e., a bandwidth of $3$ GHz~\cite[Table 5.2-1]{3gpp_NR_UE_txn}. Given that the maximum carrier bandwidth in $5$G NR is $400$ MHz, the n$257$ band could be used by at least $7$ different service providers/MOs (and more in geographies where the allotted bandwidth to each MO is less than $400$ MHz). Then, since these MOs use the same frequency range to provide services to their UEs, but using non-overlapping frequency bands, the IRS elements will reflect signals impinging on them from all MOs with similar efficiency.} 

In this context, \cite{gurgunouglu_TCOM_inter_MO_CSI_extended} studied CSI estimation in IRS-aided multiple MO systems, and~\cite{Lodro_IRS_Experimental_2022} experimentally evaluated the performance impact of IRS in the presence of multiple MOs. Further, \cite{Cai_2022_TCOM,Cai_2020_CL_practicalIRS,Xinyi_M_Imran_TVT_2024} and~\cite{Hoacheng_arxiv_2024,Schwarz_CL_2024,Wenhao_TVT_2023,Yajun_IEEE_Access,Joana_Globecom_2023} considered joint optimization of the IRS configurations and allocation of disjoint IRSs/sub-IRSs to different bands/MOs via cooperation. \textcolor{black}{However, most of these works assume that complete CSI for all links is available at all MOs, necessitating extensive inter-MO cooperation. Such cooperation between MOs is often infeasible in practice. Moreover, the precise benefit of jointly controlling all IRSs in terms of enhancing the overall network performance across multiple MOs remains largely unknown.}
Our prior works~\cite{Yashvanth_TCOM_2024,Yashvanth_WCL_2024} quantified the out-of-band performance impact of the IRS when only one MO deploys one or more IRSs, and~\cite{miridakis_WCL_2024_impactMOs} considered setups where each MO deploys its own IRS in the sub-$6$ GHz band. \textcolor{black}{Table~\ref{tab:lit_rev} presents a summary of this work and the prior studies on IRS-assisted multiple-MO systems. For the first time in the literature, we explore several aspects of IRS-aided multiple MO systems in mmWaves: quantifying the benefits of inter-MO cooperation and joint optimization, low-complexity algorithmic solutions for such cooperation, and an in-depth performance analysis with closed-form expressions. Notably, the closed-form expressions offer useful insights into the behavior and limits of such systems under various \emph{practical} transmission schemes. Further, we make these contributions without compromising the generality of the system model.} 

\begin{table*}[t!]
\color{black}
    \centering
    \caption{\textcolor{black}{Summary of literature on IRS-aided multiple MO systems.}} \label{tab:lit_rev}
    \begin{threeparttable}
    \begin{adjustbox}{width=\linewidth}
\begin{tabular}{|l|c|c|c|c|c|c|c|c|c|c|c|c|c|c|c|c|c|c|c|c|c|c|c|c|}
    \hline
   &\cite{RuiZhang_TWC_2019}  &\cite{Wang_TVT_2020} &\cite{Schober_TWC_2021}  &\cite{gurgunouglu_TCOM_inter_MO_CSI_extended} &\cite{Lodro_IRS_Experimental_2022}  &\cite{Cai_2022_TCOM} &\cite{Cai_2020_CL_practicalIRS}  &\cite{Xinyi_M_Imran_TVT_2024}  &\cite{Hoacheng_arxiv_2024}   &\cite{Schwarz_CL_2024}  &\cite{Wenhao_TVT_2023}  &\cite{Yajun_IEEE_Access}  &\cite{Joana_Globecom_2023}  & \cite{Yashvanth_TCOM_2024} &\cite{Yashvanth_WCL_2024}   &\cite{miridakis_WCL_2024_impactMOs}  & \begin{tabular}{@{}c@{}}This \\ work\end{tabular}\\ [0.5ex]
 \hline
Frequency band\tnote{$\boldsymbol{\dagger}$}  &  S$6$ & M &  S$6$ & S$6$     &S$6$      &S$6$    & S$6$  &  M & S$6$    & M  & S$6$  &  M &M  &   \!S$6$, M\!  & M  &  S$6$ &M \\ \hline
 
More than one MO? &  &  &  & \checkmark &\checkmark  &\checkmark    &    & \checkmark  & \checkmark &  \checkmark&  \checkmark&\checkmark & \checkmark  &\checkmark  &  \checkmark     &  \checkmark& \checkmark \\ \hline

Multiple UEs per MO? & {\checkmark} &  &\checkmark & &   &\checkmark   &  &  \checkmark & \checkmark  & \checkmark & \checkmark &   &  & \checkmark&  \checkmark     &  &\checkmark   \\ \hline

Multiple access scheme\tnote{$\boldsymbol{\diamond}$} & T  &--- &T, F  &---       & --- &  S & --- &S    &S  &S  & S & --- &--- &T &T   &    ---&T  \\ \hline

CSI requirement/exchange\tnote{$\boldsymbol{\#}$}  &F  & F &F &  F&  P&F     & F & F &F  &P  &P  & P&  F& P &P   &  P    &P  \\ \hline 

Multiple IRSs? & & \checkmark& &  &    &    &  &\checkmark  &  \checkmark & \checkmark &  \checkmark &  & &  &\checkmark   &         &\checkmark  \\ \hline
 
{\bf Low-complexity solution} &  & & &   &\checkmark  & &  &  &  &  \checkmark & \checkmark &  &  \checkmark&\checkmark  &\checkmark        & \checkmark  &\checkmark   \\ \hline

\begin{tabular}{@{}l@{}}{\bf MO-cooperative joint } \\ {\bf optimization schemes}\end{tabular}  &  &  &   &  & & \checkmark &\checkmark  & \checkmark &  &  &  &  &  &  &      &   &\checkmark   \\ \hline

\begin{tabular}{@{}l@{}}{\bf MO-cooperative resource} \\ {\bf \hspace{0.0cm}sharing schemes}\end{tabular} &  & & &     & &  &  & \checkmark &  &  &  & \checkmark &\checkmark  &  &      &   &\checkmark   \\ \hline

\begin{tabular}{@{}l@{}}{\bf \hspace{0.0cm}Closed form } \\ {\bf performance analysis}\end{tabular} & \checkmark &\checkmark & & \checkmark&     &  &  &    &   &   &   &  &  &\checkmark  &   \checkmark   &\checkmark  &\checkmark  \\ \hline

\begin{tabular}{@{}l@{}}{\bf Quantifying the effects} \\ {\bf of OOB IRSs}\end{tabular} &  & & & \checkmark &    &   &  &   &   &  &   &     &  &     \checkmark&   \checkmark& \checkmark &\checkmark   \\ \hline

\begin{tabular}{@{}l@{}}{\bf Quantifying the gains of} \\ {\bf inter-MO cooperation}\end{tabular} & & & &  &   &  &   &  &   &     &  &  & &     &   &   &\checkmark  \\ \hline

\end{tabular}
\end{adjustbox}
\begin{tablenotes}
  \item[$\boldsymbol{\dagger}$] M: mmWave bands; S$6$: sub-$6$ GHz bands
  \item[$\boldsymbol{\diamond}$] T: Time division multiple access; F: Frequency division multiple access; S: Space division multiple access
  \item[$\boldsymbol{\#}$] F: Full CSI required/exchanged; P: Partial CSI required/exchanged
  \end{tablenotes}
\end{threeparttable}
\end{table*}
	\vspace{-0.3cm}
	\subsection{Contributions}\label{sec:contributions}    
       To set the context, we use the following terminology: the IRS and UEs controlled/served by an MO of interest are termed \emph{in-band}, and other IRSs/UEs in the system are called \emph{out-of-band (OOB)} with respect to this same MO. We make the following key contributions in this paper:
\begin{enumerate}[leftmargin=*]
\item Considering that $2$ MOs, X and Y, control an IRS each, we derive the ergodic sum spectral-efficiency (SE) of the MOs when an overall phase at each IRS is configured as per the following implementation schemes (see Theorem~\ref{thm:theorem-1}):
\begin{enumerate}
\item \emph{Optimization with Time-sharing}: In each time slot, while an MO serves its own UE, the overall phases at the IRSs are optimized to a UE served by either MO-X or MO-Y.
\item \emph{Joint-optimization with MO cooperation}: The overall IRS phases are jointly tuned to maximize the weighted sum-SE of UEs scheduled by MOs in every time slot. 
\item \emph{No MO cooperation}: In this scheme, each MO focuses exclusively on optimizing its IRSs to ensure coherent signal reception at only its own UEs. 
\end{enumerate}
\item We show that the IRS controlled by one MO does not degrade the sum-SE of the other MO. We quantify the gain in the sum-SE of the MOs obtained {with/without OOB IRS}, and with/without cooperation (for time-sharing/joint optimization) as a function of the number of OOB IRS elements (see Theorem~\ref{theorem-2}.)
\item We next extend our results to a system where more than $2$ MOs co-exist, which deploy and control an IRS each. In particular, we derive the ergodic-sum-SE of the MOs for the above-mentioned three schemes (see Theorem~\ref{corollary:1}.)
\item Finally, even with more than $2$ MOs, we show that the OOB IRSs do not degrade the in-band performance. Further, although joint optimization/time sharing with MO cooperation still offers marginal gains relative to sum-SE when the MOs do not cooperate, the gain increases at least linearly with the number of OOB MOs. (see Theorem~\ref{corollary:2}.)
\end{enumerate}
		Our results explicitly uncover the dependence of the ergodic sum-SE of the MOs in IRS-aided mmWave systems on system parameters such as the number of IRS elements, in-band and OOB cascaded channel paths, SNR of operation, etc. 
		
We numerically validate our analytical results and illustrate that joint optimization/cooperation among MOs provides marginal gains compared to when an MO configures its IRS without any cooperation.  For instance, with $2$-MOs, each with an $N=16$-element IRSs, the gain in the SE of an MO obtained at $80$ dB transmit SNR via joint optimization and cooperation over a no-cooperation policy is about $2\%$. Also, this improvement monotonically decreases with the number of OOB IRS elements and transmit SNR; for e.g., it is $0.4\%$ and $0.08\%$ for $N=32,64$, respectively. As a result, cooperation between MOs to optimize the IRSs may not be needed in IRS-aided mmWave systems with multiple MOs. Each MO can deploy and control its IRS independently, and the IRS of one MO does not degrade the performance of another MO. 
		
		\indent \textcolor{black}{\emph{Notation:} 
$|\cdot|,\angle\cdot$ denote the magnitude and phase of a complex number (vector); ${(\cdot)^*}$ stands for complex conjugation; $\mathbbm{1}_{\{\cdot\}}$ is the indicator function; $|\mathcal{A}|$ is the cardinality of set $\mathcal{A}$; $A \stackrel{d}{=} B$ means $A$ and $B$ have the same distribution; $\odot$ is the Hadamard product; $\mathcal{U}$ denotes uniform distribution; $\mathcal{CN}$ denotes circularly symmetric complex Gaussian distribution; $\text {Bin}$ denotes binomial distribution; $\Re(\cdot)$ and $\mathbb{R}^+$ are the real part and set of positive real numbers; ${ \sf{Pr}}(\cdot)$ and $\langle \cdot 
\rangle \triangleq \mathbb{E}[\cdot]$ refer to the probability measure and expectations. $\mathcal{O}(\cdot)$ and $\Gamma(\cdot)$ denote the Landau's Big-O function and the Gamma function.} \!\!
\vspace{-0.3cm}
\section{System Model and Problem Description}\label{sec:system_model}
    Multiple MOs operating over different and non-overlapping mmWave bands co-exist in a given geographical area and provide services to the UEs subscribed to them. \textcolor{black}{For mathematical brevity, we describe the model for a system with two MOs, say X and Y, but the model directly extends to any number of MOs, as we describe in Sec.~\ref{sec:M-BS & M-IRS System}.}  
     The MOs X and Y operate over non-overlapping frequency bands, and both use time-division multiple access (TDMA) to serve one of $K$ UEs on a frequency band centered at $f_1$ and one of $Q$ UEs on a frequency band centered at $f_2$, respectively, in each time slot. 
     Also, their base stations,\footnote{\textcolor{black}{For simplicity, we use single antenna BSs in this work, similar to~\cite{Cai_2020_CL_practicalIRS,Schwarz_CL_2024}. However, our results can be extended to multiple antenna cases also.}} BS-X and BS-Y, deploy and control an $N_1$-element IRS-X and an $N_2$-element IRS-Y, respectively. Due to the high attenuation in the mmWaves, the direct links between the BSs and UEs are blocked {\cite{Lin_TCOM_2022_Ch_Est_IRS,Cao_ISPIMRC_2020}}.
     The downlink signal received at UE-$k$, served by MO-X, is given by
     \color{black}
        \begin{equation} \label{eq:kth_downlink_signal}
        y_{k} = \left(\mathbf{g}_{\text{X}k}^T\boldsymbol{\Theta_1}\mathbf{f}_{\text{XX}} + \mathbf{g}_{\text{Y}k}^T\boldsymbol{\Theta_2}\mathbf{f}_{\text{XY}}\right)x_k + n_k,
    \end{equation} 
    where $\mathbf{g}_{\text{X}k} \in \mathbb{C}^{N_1}$ and $\mathbf{g}_{\text{Y}k} \in \mathbb{C}^{N_2}$ are the channels from IRS-X and IRS-Y to UE-$k$, respectively; $\mathbf{f}_{\text{XX}} \in \mathbb{C}^{N_1}$ and $\mathbf{f}_{\text{XY}} \in \mathbb{C}^{N_2}$ are the channels from BS-X to IRS-X and IRS-Y, respectively, 
     \color{black}
    $x_k$ is the information symbol for UE-$k$ with average power constraint $\mathbb{E}[|x_k|^2] \leq P$ and $n_k\sim \mathcal{CN}(0,\sigma^2)$ is the additive noise at UE-$k$.  
    Similarly, the downlink signal from BS-Y to UE-$q$, served by MO-Y, can be written as
    \color{black}
    \begin{equation}\label{eq:qth_downlink_signal}
        y_{q} = \left(\mathbf{t}_{\text{X}q}^T\boldsymbol{\Theta_1}\mathbf{f}_{\text{YX}} +\mathbf{t}_{\text{Y}q}^T\boldsymbol{\Theta_2}\mathbf{f}_{\text{YY}}\right)x_q + n_q,
    \end{equation} where $\mathbf{t}_{\text{X}q} \in \mathbb{C}^{N_1}$ and $\mathbf{t}_{\text{Y}q} \in \mathbb{C}^{N_2}$ are the channels from IRS-X and IRS-Y to UE-$q$, respectively, $\mathbf{f}_{\text{YX}} \in \mathbb{C}^{N_1}$ and $\mathbf{f}_{\text{YY}} \in \mathbb{C}^{N_2}$ are the channels from BS-Y to IRS-X and IRS-Y, respectively, 
    \color{black}
    $x_q$ is the information symbol for UE-$q$ with power constraint $\mathbb{E}[|x_q|^2] \leq P$ and $n_q\sim \mathcal{CN}(0,\sigma^2)$ is the additive noise at UE-$q$. \textcolor{black}{In particular, $\mathbf{f}_{ab}$ denotes the channel from BS-$a$ to IRS-$b$, $\mathbf{g}_{cd}$ denotes the channel from IRS-$c$ to the $d$th UE served by MO-X, and $\mathbf{t}_{cd}$ denotes the channel from IRS-$c$ to the $d$th UE served by MO-Y.}
     Finally, $\boldsymbol{\Theta_1} \in \mathbb{C}^{N_1 \times N_1}$ and $\boldsymbol{\Theta_2} \in \mathbb{C}^{N_2 \times N_2}$ are diagonal matrices with unit modulus reflection coefficients of IRS-X and IRS-Y, respectively. Figure~\ref{System_Model_2_Operator_2_IRS} illustrates our system model. 
    
   \indent \emph{Terminology:} Since MO-X configures IRS-X to serve UE-$k$, we refer to the IRS-X and UE-$k$ as the \emph{in-band} IRS and UE, respectively, from MO-X's viewpoint. Similarly, the IRSs or UEs that are not controlled/served by the BS-X (operating on a different band) are \emph{out-of-band (OOB)} nodes from MO-X's viewpoint. Further, the link from MO-X to UE-$k$ via IRS-X is the in-band channel; the links from MO-X to UE-$k$ via OOB IRSs are OOB channels. These apply to other MOs also.
   \vspace{-0.1cm}
\begin{figure}[t!]
			\vspace{-0.1cm}
		\centering
		\includegraphics[width=\linewidth]{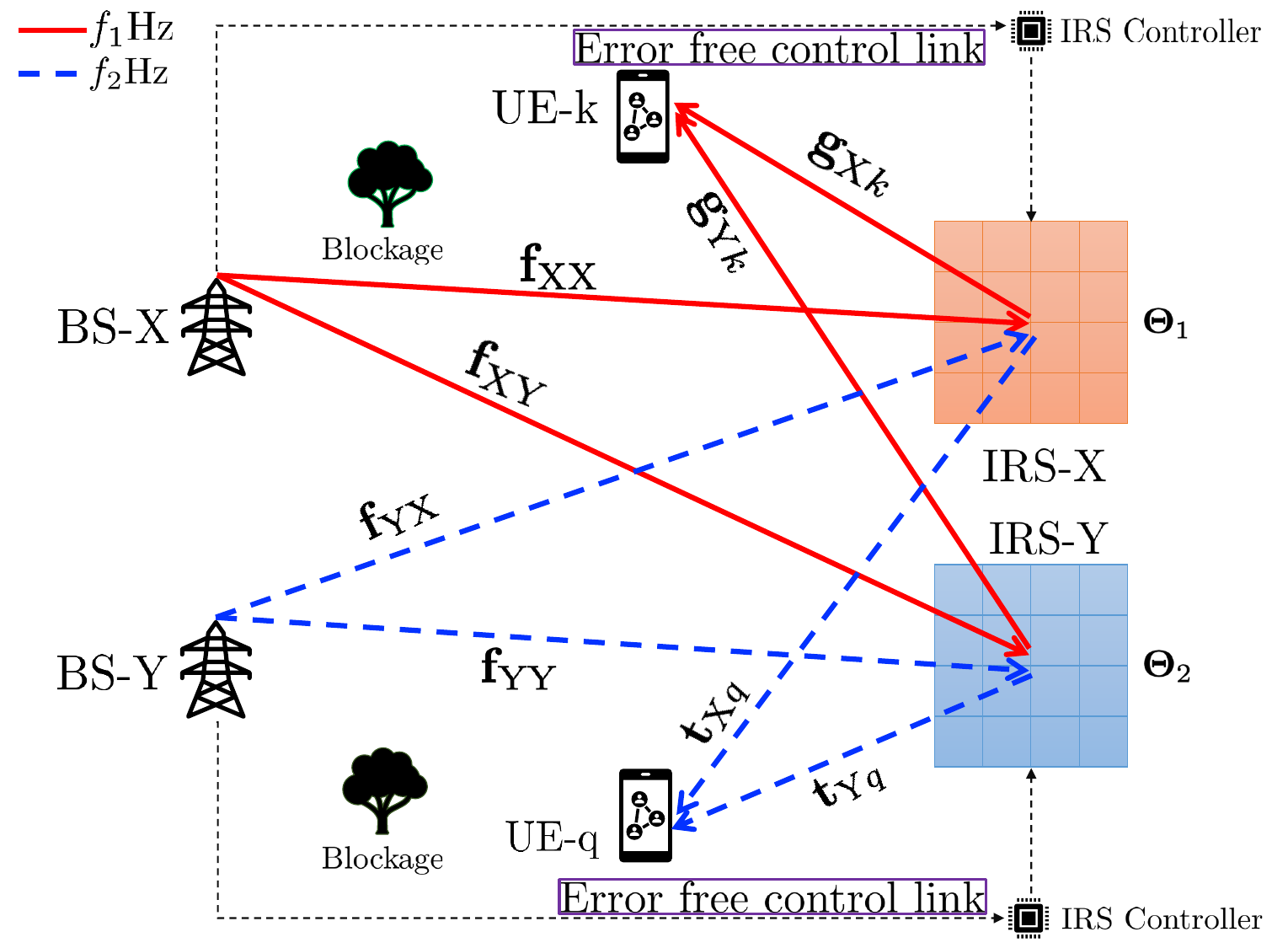}		
		\caption{\textcolor{black}{System Model of 2-BS \& 2-IRS system.}}
		\label{System_Model_2_Operator_2_IRS}
	\end{figure}
\subsection{Channel Model}
 We use the standard Saleh-Venezuela (SV) model to describe the channels in the mmWave frequency bands \textcolor{black}{\cite{Wang_TWC_2022_IRS_BA,Lin_TCOM_2022_Ch_Est_IRS}}. The channel from BS-$Z$ to IRS-$W$ ($Z, W \in \{ \text{X}, \text{Y}\}$) is
    \begin{equation}\label{eq:ch_model_mmwave_double_IRS_2}
    	\vspace{-0.2cm}
	\color{black}
        \mathbf{f}_{ZW} = \sqrt{N_p\Big/L^{(p)}_{Z}}\sum\nolimits_{l=1}^{L^{(p)}_{Z}}\gamma^{(p)}_{l,Z} \mathbf{a}^*_{N_p}(\phi^{(p)}_{l,Z}),
    \end{equation} 
    \color{black}
    where  $p = 1 \cdot \mathbbm{1}_{\{W = \text{X}\}} + 2\cdot \mathbbm{1}_{\{W = \text{Y}\}}$, 
    $L^{(p)}_{Z}$ is the number of resolvable paths from BS-$Z$ to the IRS-$W$ and $\phi^{(p)}_{l,Z}$ is the sine of the angle of arrival of the $l$th path from BS-$Z$ to IRS-$W$. Similarly, the channel from IRS-$W$ to UE-$r$ served by MO-X/MO-Y is given by
    	\vspace{-0.1cm}
	\color{black}
    \begin{equation}\label{eq:ch_model_mmwave_double_IRS_1}
    		\vspace{-0.2cm}
       \mathbf{g}_{Wr} \Big/ \mathbf{t}_{Wr} = \sqrt{N_p\Big/L^{(p)}_{r}} \sum\nolimits_{l=1}^{L^{(p)}_{r}}\gamma^{(p)}_{l,r} \mathbf{a}^*_{N_p}(\psi^{(p)}_{l,r}),
    \end{equation} 
    \color{black}
    where $L^{(p)}_{r}$ is the number of resolvable paths from IRS-$W$ to UE-$r$, $\psi^{(p)}_{l,r}$ is the sine of angle of departure of $l$th path from IRS-$W$ to UE-$r$, and $N_p$ is the number of IRS elements in IRS-$W$. The sine terms are sampled from an appropriate distribution $\mathcal{P}_{\mathcal{A}}$ which is discussed in the sequel. The fading coefficients, $\gamma^{(p)}_{l,Z}$ and $\gamma^{(p)}_{l,r}$ are independently sampled from $\mathcal{CN}(0,\beta^{(p)}_{Z})$ and $\mathcal{CN}(0,\beta^{(p)}_{r})$, respectively, where $\beta^{(p)}_{Z}$ and $\beta^{(p)}_{r}$ denote the path loss in BS-IRS and IRS-UE links, respectively. 
       Finally, we consider an $N$-element uniform linear array (ULA) based IRS,\footnote{Similar results as in this paper can be easily obtained for other array types such as planar arrays also.} similar to~\cite{Wei_Yu_TWC_2023_Sparse_Ch_estim_IRS}, with half-wavelength inter-element spacing; its array response vector $\mathbf{a}_{N}(\psi)$ is 
    \begin{equation} \label{eq:array_response}
        \mathbf{a}_N(\psi) =\dfrac{1}{\sqrt{N}} [1, e^{-j\pi\psi},\ldots,e^{-j(N-1)\pi\psi}]^T \in \mathbb{C}^{N}.
  \vspace{-0.15cm}
	    \end{equation}
\subsubsection{Cascaded Channel Representation}
Substituting the expressions for individual channels given in~\eqref{eq:ch_model_mmwave_double_IRS_2}, \eqref{eq:ch_model_mmwave_double_IRS_1} into~\eqref{eq:kth_downlink_signal}, the channel at UE-$k$ can be simplified as in~\eqref{kth_channel_eq}, \eqref{kth_channel_eq_2} on the top of the page, 
\begin{figure*}
	\vspace{-0.85cm}
	\begin{align}
		h_k  & = \frac{N_1}{\sqrt{L_{k1}}}\sum_{l=1}^{L_{k1}}\gamma^{(1)}_{l,X}\gamma^{(1)}_{l,k}\mathbf{a}_{N_1}^H(\psi^{(1)}_{l,k})\boldsymbol{\Theta_1}\mathbf{a}_{N_1}^*(\phi^{(1)}_{l,X})  
		+\frac{N_2}{\sqrt{L_{k2}}}\sum_{l=1}^{L_{k2}}\gamma^{(2)}_{l,X}\gamma^{(2)}_{l,k}\mathbf{a}_{N_2}^H(\psi^{(2)}_{l,k})\boldsymbol{\Theta_2}\mathbf{a}_{N_2}^*(\phi^{(2)}_{l,X}) \label{kth_channel_eq}  \\ 
		&\stackrel{(a)}{=} \frac{N_1}{\sqrt{L_{k1}}}\sum_{l=1}^{L_{k1}} \gamma^{(1)}_{l,X}\gamma^{(1)}_{l,k}\left(\mathbf{a}_{N_1}^H(\phi^{(1)}_{l,X})\odot\mathbf{a}_{N_1}^H(\psi^{(1)}_{l,k})\right)\boldsymbol{\theta_1} 
		+ \frac{N_2}{\sqrt{L_{k2}}}\sum_{l=1}^{L_{k2}} \gamma^{(2)}_{l,X}\gamma^{(2)}_{l,k}\left(\mathbf{a}_{N_2}^H(\phi^{(2)}_{l,X})\odot\mathbf{a}_{N_2}^H(\psi^{(2)}_{l,k})\right)\boldsymbol{\theta_2} \label{kth_channel_eq_2},  
	\end{align}
	\vspace{-0.2cm}
	\hrule
	\vspace{-0.6cm}
\end{figure*}
where $L_{k1} \triangleq L_{X}^{(1)}L_{k}^{(1)} $ is the number of resolvable in-band paths from BS-X to UE-$k$ through IRS-X (\textcolor{black}{see Sec.~\ref{sec_angle_book_dist} for details on resolvability}), $L_{k2} \triangleq L_{X}^{(2)}L_{k}^{(2)}$ is the number of resolvable OOB paths from BS-X to UE-$k$ through IRS-Y, $\boldsymbol{\theta_1}=\text{diag}(\boldsymbol{\Theta_1})\in\mathbb{C}^{N_1}$, $\boldsymbol{\theta_2}=\text{diag}(\boldsymbol{\Theta_2})\in\mathbb{C}^{N_2}$, and $(a)$ is obtained using the properties of the Hadamard product. The first and second terms in~\eqref{kth_channel_eq}, \eqref{kth_channel_eq_2} represent the effective channels through the in-band IRS-X and the OOB IRS-Y, respectively. Since the Hadamard product of two array vectors is also an array vector aligned to a different angle, we have
\vspace{-0.3cm}
\begin{multline}\label{kth_channel_eq_1}
	h_k  = \frac{N_1}{\sqrt{L_{k1}}}\sum_{l=1}^{L_{k1}} \gamma^{(1)}_{l,X}\gamma^{(1)}_{l,k}\mathbf{\dot{a}}_{N_1}^H(\omega^{(1)}_{X,k,l})\boldsymbol{\theta_1}  \\
	+ \frac{N_2}{\sqrt{L_{k2}}}\sum_{l=1}^{L_{k2}} \gamma^{(2)}_{l,X}\gamma^{(2)}_{l,k}\mathbf{\dot{a}}_{N_2}^H(\omega^{(2)}_{X,k,l})\boldsymbol{\theta_2},
	\vspace{-0.3cm}
\end{multline}
\textcolor{black}{where $\omega^{(1)}_{X,k,l} \triangleq \sin_{(p)}^{-1}\Bigl(\sin(\phi^{(1)}_{l,X}) + \sin(\psi^{(1)}_{l,k})\Bigl)$, and $\omega^{(2)}_{X,k,l} \triangleq \sin_{(p)}^{-1}\Bigl(\sin(\phi^{(2)}_{l,X}) + \sin(\psi^{(2)}_{l,k})\Bigl)$ denote the cascaded channel at UE-$k$ from MO-X via IRS-X, and IRS-Y, respectively, in the $l$th path}. Here, $\sin_{(p)}^{-1}(x)$ is defined so that $x \in [-1,1)$, the principal argument of the inverse sine function~\cite{Wei_Yu_TWC_2023_Sparse_Ch_estim_IRS},~\cite[Eq.~$32$]{Yashvanth_TCOM_2024}.
Further, $\mathbf{\dot{a}}_{N}(\omega)\triangleq\frac{1}{\sqrt{N}}\mathbf{a}_{N}(\omega)$, with $\mathbf{a}_N(\omega)$ as defined in \eqref{eq:array_response}. Similarly, the channel at UE-$q$ is 
\vspace{-0.3cm}
\begin{multline}\label{qth_channel_eq_1}
\vspace{-0.1cm}
	h_q = \frac{N_1}{\sqrt{L_{q1}}}\sum_{l=1}^{L_{q1}} \gamma^{(1)}_{l,Y}\gamma^{(1)}_{l,q}\mathbf{\dot{a}}_{N_1}^H(\omega^{(1)}_{Y,q,l})\boldsymbol{\theta_1}  \\
	+ \frac{N_2}{\sqrt{L_{q2}}}\sum_{l=1}^{L_{q2}} \gamma^{(2)}_{l,Y}\gamma^{(2)}_{l,q}\mathbf{\dot{a}}_{N_2}^H(\omega^{(2)}_{Y,q,l})\boldsymbol{\theta_2},
	\vspace{-0.3cm}
\end{multline} 
where  $L_{q1} \triangleq L_{Y}^{(1)}L_{q}^{(1)} $ is the number of resolvable OOB paths from BS-Y to UE-$q$ via IRS-X; $L_{q2} \triangleq L_{Y}^{(2)}L_{q}^{(2)}$ is the number of resolvable in-band paths from BS-Y to UE-$q$ via IRS-Y. 
\subsubsection{\textcolor{black}{Angle Distribution}}\label{sec_angle_book_dist}
    We now explain the distribution of the cascaded angles specified in~\eqref{kth_channel_eq_1},~\eqref{qth_channel_eq_1}. 
        Since an $N$-element ULA can form at most $N$ \emph{resolvable beams}~\cite{Gen_TWC_2016,Wang_TWC_2022_IRS_BA}, the paths with angular separations smaller than the Rayleigh resolution limit, i.e., $2 \pi / N$ radians, are unresolvable and appear as a single path with an appropriate fading coefficient. 
     To that end, we define the set of resolvable beams formed by the IRS as 
    \begin{equation*}
	\mathcal{A} \triangleq \left\{ \mathbf{a}_N(\omega), \omega \in \mathbf{\Omega}\right\}\!; \mathbf{\Omega} \triangleq\! \left\{\! \left(\!-1 + \!\frac{2i}{N}\right)\!\bigg\rvert i =  0,\ldots,N-1\!\right\}\!,\!
\end{equation*} 
where $\mathbf{\Omega}$ is the \emph{resolvable anglebook} of the IRS. Then, we model its distribution $\mathcal{P}\!_{\mathcal{A}}$ by a uniform distribution:
	\vspace{-0.1cm}
\begin{equation}\label{eq:codeboook_ditbn}
		\vspace{-0.1cm}
	\mathcal{P}_{\mathcal{A}}(\omega) = ({1}/{\left| \mathbf{\Omega}\right|})  \cdot \mathbbm{1}_{\left\{\omega \in \mathbf{\Omega}\right\}} = ({1}/{N})\cdot   \mathbbm{1}_{\left\{\omega \in \mathbf{\Omega}\right\}}.
\end{equation} 
Hence, we sample all the cascaded angles, $\{\omega^{(1)}_{X,k,l}\}_l$, $\{\omega^{(2)}_{X,k,l}\}_l$, $\{\omega^{(1)}_{Y,q,l}\}_l$, $\{\omega^{(2)}_{Y,q,l}\}_l$ from $\boldsymbol{\Omega}$ given above, similar to~\cite{Wei_Yu_TWC_2023_Sparse_Ch_estim_IRS}. Furthermore, since IRS-X(Y) forms at most $N_1$ ($N_2$) resolvable paths, we have $L_{k1}, L_{q1}\leq N_1$;  $L_{k2}, L_{q2}\leq N_2$.

\subsection{Choice of IRS Configurations}
Recall that MO-X controls the (in-band) IRS-X to optimally serve its UEs, while BS-X cannot directly control the (OOB) IRS-Y. Similarly, MO-Y controls the (in-band) IRS-Y, and (OOB) IRS-X is not directly controllable by BS-Y. In such a scenario, at any instant in time, one of the cascaded in-band paths (in \eqref{kth_channel_eq_1} and \eqref{qth_channel_eq_1}) 
contains the maximum energy, and aligning the in-band IRSs to that path will procure near optimal benefits~\cite{Wang_TVT_2020}.
Without loss of generality, \textcolor{black}{we label the strongest path as the first in-band path}. Then, the strongest in-band cascaded path of UE-$k$ is $h_{k,1} \triangleq \frac{N_1}{\sqrt{L_{k1}}} \gamma^{(1)}_{1,X}\gamma^{(1)}_{1,k}\mathbf{\dot{a}}_{N_1}^H(\omega^{(1)}_{X,k,1})\boldsymbol{\theta_1}$. Further, recall that each MO prioritizes optimizing its IRS to align it along the in-band channel at its scheduled UE. In particular, since MO-X controls $\boldsymbol{\theta_1}$, to maximize the the \textcolor{black}{channel gain} $|h_{k,1}|^2$, by using Cauchy-Schwartz (CS) inequality, the $n$th entry of the optimal IRS configuration vector $\boldsymbol{\theta_1}^{\mathrm{opt}}$  is $\theta_{1,n} = e^{j\phi_1} e^{j\left(-\angle \gamma^{(1)}_{1,X}-\angle \gamma^{(1)}_{1,k}-\pi(n-1)\omega^{(1)}_{X,k,1} \right)}$, where $\phi_1$ is an overall phase angle applied to IRS-X which still preserves the optimality.\footnote{For e.g., it can be chosen to phase-align the channel $h_{k,1}$ with the overall virtual ``direct path'' formed by the cascaded channel through IRS-Y, i.e., with the phase of the second term in \eqref{kth_channel_eq_1}. We will explain this in the sequel.} Similarly, we can obtain the optimal configuration for IRS-Y that maximizes $|h_{q,1}|^2$. Thus, the optimal IRS phase vectors can be written compactly as~\cite{Yashvanth_TCOM_2024} 
    \vspace{-0.1cm}
    \begin{align}
    	   	\vspace{-0.1cm}
		\boldsymbol{\theta_1}\!^{\mathrm{opt}} &= \frac{{\gamma^{(1)*}_{1,X}}\:{\gamma^{(1)*}_{1,k}}}{\left|\gamma^{(1)}_{1,X}\gamma^{(1)}_{1,k} \right|} \times N_1 \mathbf{\dot{a}}_{N_1}(\omega^{(1)}_{X,k,1}) \times e^{j\phi_1}, \label{eq:optimum_vector_k}\\
		\boldsymbol{\theta_2}\!^{\mathrm{opt}} &= \frac{{\gamma^{(2)*}_{1,Y}}\:{\gamma^{(2)*}_{1,q}}}{\left|\gamma^{(2)}_{1,Y}\gamma^{(2)}_{1,q} \right|} \times N_2 \mathbf{\dot{a}}_{N_2}(\omega^{(2)}_{Y,q,1})\times e^{j\phi_2},\label{eq:optimum_vector_q}
    \end{align} 
respectively, and the choice of $\phi_1$, $\phi_2$ will be explained next.

\begin{table*}[t!]
\color{black}
\centering
\caption{\textcolor{black}{Notation}.}
\label{oob_work_multiple_IRS_Tab:table_acronym_paper}
\resizebox{\textwidth}{!}{%
\begin{tabular}{ | m{1.2cm} | l | m{1.2cm} | l |}
\hline
\!\!\textbf{Variable} & \textbf{Definition}                & \!\!\textbf{Variable} & \textbf{Definition}              \\ \hline
$N_1\big/N_2$    & Number of elements in IRS-X $\big/$ IRS-Y  & $\gamma_{l,Z}^{(p)}$    & Gain of the $l$th path from BS-Z to IRS-$p$          \\ \hline
$M$    & Total number of MOs  &     $\gamma_{l,r}^{(p)}$    & Gain of the $l$th path from IRS-$p$ to UE-$r$    \\ \hline
$\boldsymbol{\Theta_1}\big/\boldsymbol{\Theta_2}$   & Phase matrix at IRS-X $\big/$ IRS-Y  &   \!\!\!$\beta_{X,k}^{(1)}\big/\beta_{X,k}^{(2)}$ & \begin{tabular}[c]{@{}l@{}} Path loss in the BS-X to UE-$k$ link via \\ IRS-X $\big/$ IRS-Y\end{tabular}      \\ \hline
$K$      &  Number of UEs served by BS-X              & \!\!\!$\beta_{Y,q}^{(1)}\big/\beta_{Y,q}^{(2)}$ & \begin{tabular}[c]{@{}l@{}} Path loss in the BS-Y to UE-$q$ link via \\ IRS-X $\big/$ IRS-Y \end{tabular}    \\ \hline
$Q$     &  Number of UEs served by BS-Y               & $\mathbf{f}_{ZW}$    & Channel from BS-Z to IRS-W\\ \hline
$L_{k1}\big/L_{k2}$    & \begin{tabular}[c]{@{}l@{}} Number of resolvable paths from BS-X \\to UE-$k$ through IRS-X $\big/$ IRS-Y  \end{tabular}          &   $\!\!\!\mathbf{g}_{Wr} \Big/ \mathbf{t}_{Wr}$   & \begin{tabular}[c]{@{}l@{}} Channel from IRS-W to UE-$r$ \\ served by MO-X $\big/$ MO-Y  \end{tabular} \\ \hline
$L_{q1}\big/L_{q2}$    & \begin{tabular}[c]{@{}l@{}} Number of resolvable paths from BS-Y \\to UE-$q$ through IRS-X $\big/$ IRS-Y  \end{tabular} &  $h_k\big/h_q$ & \begin{tabular}[c]{@{}l@{}} Overall channel from BS-X to UE-$k$\\ $\big/$ BS-Y to UE-$q$ \end{tabular}  \\ \hline
$\phi_1\big/\phi_2$    & \begin{tabular}[c]{@{}l@{}}  Overall phase shift applied \\ at IRS-X $\big/$ IRS-Y by BS-X $\big/$ BS-Y \end{tabular}      &  $\zeta$    & \begin{tabular}[c]{@{}l@{}}  Fraction of time slots to optimize the \\ IRSs to UE-$k$ served by MO-X \end{tabular}    \\ \hline
$\omega^{(1)}_{X,k,l}$ $\big/$ $\omega^{(2)}_{X,k,l}$     & \begin{tabular}[c]{@{}l@{}}Cascaded normalized angle of $l$th path  \\ from BS-X to UE-$k$ via IRS-X $\big/$ IRS-Y\end{tabular}            &  $\omega^{(1)}_{Y,q,l}$ $\big/$ $\omega^{(2)}_{Y,q,l}$    & \begin{tabular}[c]{@{}l@{}} Cascaded normalized angle of $l$th path \\ from BS-Y to UE-$q$ via IRS-X $\big/$ IRS-Y \end{tabular}  \\ \hline
$P$     & Transmit power at the BSs         &   $\text{CO}$  & \begin{tabular}[c]{@{}l@{}} Boolean parameter to indicate whether \\ we allow inter-MO cooperation\end{tabular}  \\ \hline
$\sigma^2$     & Noise variance at the UEs        &  $\langle R_{\text{X}} \rangle$ $\big/$ $\langle R_{\text{Y}} \rangle$    & \begin{tabular}[c]{@{}l@{}}Achievable ergodic sum-SEs  \\ of MOs X $\big/$ Y\end{tabular}   \\ \hline
\end{tabular}
}
\end{table*}

\begin{remark}
The IRS configurations in~\eqref{eq:optimum_vector_k} and~\eqref{eq:optimum_vector_q} do not require knowledge of the channel through the OOB IRS and hence are scalable for any number of MOs. Notably, even in the absence of OOB MOs, the in-band IRS associated with the MO will still procure an SNR that scales quadratically in the number of IRS elements. Consequently, the goal of this paper to demonstrate the utility of choosing the overall phase shifts $\phi_1$ and $\phi_2$ via cooperation rather than cooperatively optimizing the complete IRS phase vectors. 
\end{remark}
\begin{remark}
In order for BS-X and BS-Y to configure IRS-X and IRS-Y to phase values given in~\eqref{eq:optimum_vector_k} and~\eqref{eq:optimum_vector_q}, respectively, it is necessary for both BSs to acquire the knowledge of the respective in-band CSIs at their in-band UEs-$k$ and $q$, respectively. A straightforward approach to achieve this is via inter-MO cooperation during channel estimation (CE), as follows: when one MO performs in-band CE through its IRS, the OOB IRSs are turned off to prevent inter-MO pilot and IRS contamination~\cite{gurgunouglu_TCOM_inter_MO_CSI_extended}. 
This allows all MOs to configure their IRSs according to~\eqref{eq:optimum_vector_k} and~\eqref{eq:optimum_vector_q}. We note that designing and analyzing the feasibility of practical CE protocols in IRS-aided multiple MO systems remains an open problem. However, since our goal is to characterize the impact of multiple MOs deploying IRSs on each other's achievable data rates, we do not account for these overheads in our analysis.
\end{remark}
\subsection{Problem Statement}\label{sec:Prob_statement}
In a $2$-MO system, as shown in Fig.~\ref{System_Model_2_Operator_2_IRS}, each IRS will reflect the signals transmitted by both the MOs. Then, the SE achieved by UE-$k$ scheduled by MO-X at time slot $t$ is 
\vspace{-0.1cm}
\color{black}
\begin{equation*}
\vspace{-0.1cm}
R_k(t) = \log_2\left(1+\frac{P}{\sigma^2} \left| \mathbf{g}_{\text{X}k}^T\boldsymbol{\Theta_1}(t)\mathbf{f}_{\text{XX}}+ \mathbf{g}_{\text{Y}k}^T\boldsymbol{\Theta_2}(t)\mathbf{f}_{\text{XY}}\right|^2\right),
\vspace{-0.1cm}
\end{equation*}
\color{black}
and the SE achieved by UE-$q$ scheduled by MO-Y is 
\vspace{-0.1cm}
\color{black}
\begin{equation*}
\vspace{-0.1cm}
R_q(t) = \log_2\left(1 + \frac{P}{\sigma^2} \left|\mathbf{t}_{\text{X}q}^T\boldsymbol{\Theta_1}(t)\mathbf{f}_{\text{YX}}\!\! + \mathbf{t}_{\text{Y}q}^T\boldsymbol{\Theta_2}(t)\mathbf{f}_{\text{YY}}\right|^2\right),
\vspace{-0.1cm}
\end{equation*} 
\color{black}
where $\boldsymbol{\Theta}_1(t)$ and $\boldsymbol{\Theta}_2(t)$ are set using $\boldsymbol{\theta}_1^{\mathrm{opt}}$ and $\boldsymbol{\theta}_2^{\mathrm{opt}}$ as given in~\eqref{eq:optimum_vector_k} and~\eqref{eq:optimum_vector_q}, respectively for the UEs scheduled in time slot $t$. However, note that the choice of overall phase shifts $\phi_1$ and $\phi_2$ still offers flexibility in terms of being able to combine signals at UEs across both in-band and OOB IRSs. In this context, we consider the following scenarios:
\begin{enumerate}[leftmargin=*]
\item \emph{Joint optimization of IRSs with MO cooperation}: Here, the MOs cooperate to jointly tune the overall phase shifts at the IRSs in every time slot $t$ to maximize the weighted sum-SE of the scheduled UEs.
Mathematically, the problem is
  \vspace{-0.1cm}
\begin{equation}
\!\!\!\!\!\!\!{\phi_1}^{\mathrm{opt}}(t),{\phi_2}^{\mathrm{opt}}(t)= \argmax_{{\phi_1}(t) ,{\phi_2}(t) } w_k R_k(t) +  w_q R_q(t),\tag{P1} 
\vspace{-0.1cm}
\end{equation}
where $w_k,w_q$ are the weights associated with acheivable SEs of UEs $k$, $q$ of MO-X and Y, respectively. 
\item \emph{Optimization of IRSs with \textcolor{black}{time-sharing}}:  Here,  a subset (denoted by $\mathcal{T}_X$) of the time slots are used by MO-$X$ to configure the overall phase shifts of both IRSs to maximize the SE of UE-$k$, and the remaining time slots (denoted by $\mathcal{T}_Y$) are used by MO-$Y$ to optimize the overall phase shifts at the IRSs for UE-$q$. Mathematically, in every time slot $t$, 
  \begin{equation}\label{time-spliiting_optimization problem}
\vspace{-0.1cm}
\hspace{-0.4cm} {\phi_1}^{\mathrm{opt}}(t),{\phi_2}^{\mathrm{opt}}(t)= \argmax_{{\phi_1} (t),{\phi_2}(t) } \sum\limits_{i \in {\{X,Y\}}} \!\!\!R_i(t)\mathbbm{1}_{\{t \in \mathcal{T}_i\}}. \tag{P2}
\end{equation}

\item \emph{Optimization of IRSs without MO cooperation}: Here, the two MOs optimize only their own IRSs to maximize the SE of their UEs (by ignoring the presence of an IRS deployed by another MO.) Mathematically, we realize this by setting: 
\begin{equation}\label{eq_scheme_3_overall_phase}
\phi_1^{\mathrm{opt}}(t) = \phi_2^{\mathrm{opt}}(t) = 0, \ \forall t.
\end{equation}
\end{enumerate}

Then, we answer the following questions:
    \begin{itemize}[leftmargin=*]
        \item How does the ergodic SE of the MOs \textcolor{black}{scale with the system parameters} in the three cases?
         \item Does the presence of an OOB IRS degrade the performance of a given MO?
 \item What is the value of cooperation between the MOs in terms of the achievable ergodic SE?
        \item How do the above answers extend to $M>2$ MO-systems?
    \end{itemize}
   We answer these questions in the following sections. 
   \vspace{-0.15cm}
\section{Performance Analysis in a $2$-MO System} \label{sec:optimization_mmWave_2-BS}
    In this section, we analyze the achievable ergodic sum-SE of the $2$-MO system described above.     
 We first make the following observations about the IRS configurations in~\eqref{eq:optimum_vector_k} and~\eqref{eq:optimum_vector_q}:
     \begin{enumerate}[leftmargin=*]
      \item The IRS vectors are directional in nature and point to the angle of the channel to which it is optimized.
     \item Although the IRS vector $\boldsymbol{\theta}_1\!^{\mathrm{opt}}$ aligns to the in-band path at UE-$k$, it is a random phasor from the UE-$q$'s viewpoint. Similarly, $\boldsymbol{\theta}_2\!^{\mathrm{opt}}$ is optimal to UE-$q$'s in-band path, and is randomly configured from UE-$k$'s viewpoint. 
         \end{enumerate}
     From these observations, IRS-X aligns with the channel to UE-$q$ with probability  $\frac{{L}_{q1}}{N_1}$ and it does not contribute to the channel at UE-$q$ with probability $1-\frac{{L}_{q1}}{N_1}$~\cite[Proof of Theorem $3$]{Yashvanth_TCOM_2024}. Similarly, IRS-Y contributes to the channel at UE-$k$ with probability $\frac{{L}_{k2}}{N_2}$ and does not align with UE-$k$ with probability $1-\frac{{L}_{k2}}{N_2}$. Based on these, four different events arise as summarized in Fig.~\ref{fig:System_Model_with-Events}. \textcolor{black}{Hence, the overall achievable performance in a $2$-MO system is determined by the choice of the overall phase shifts, $\phi_1$ and $\phi_2$, used in each of these four events.\footnote{These events correspond to the OOB effect of the IRSs; by \eqref{eq:optimum_vector_k} and \eqref{eq:optimum_vector_q}, each IRS is always aligned to the in-band UE's channel from its BS.} We next analyze the performances of the different schemes listed in Sec.~\ref{sec:Prob_statement} with varying degrees of cooperation between the MOs under these events.}
   \begin{figure}[t]
 	\vspace{-0.2cm}
 	\centering
 	\includegraphics[width=\linewidth]{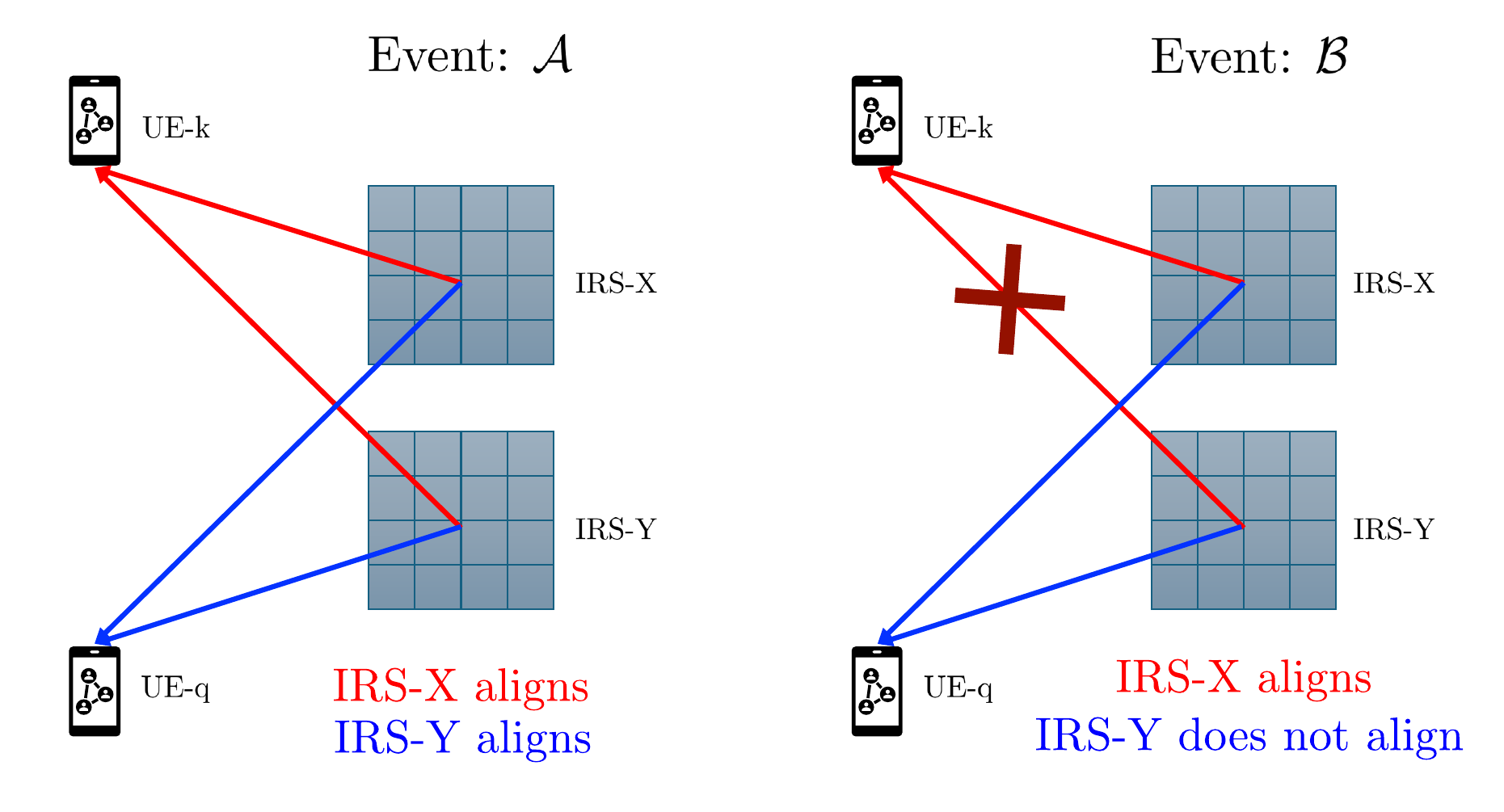}
		\newline
	\vspace{-0.2cm}
	\includegraphics[width=\linewidth]{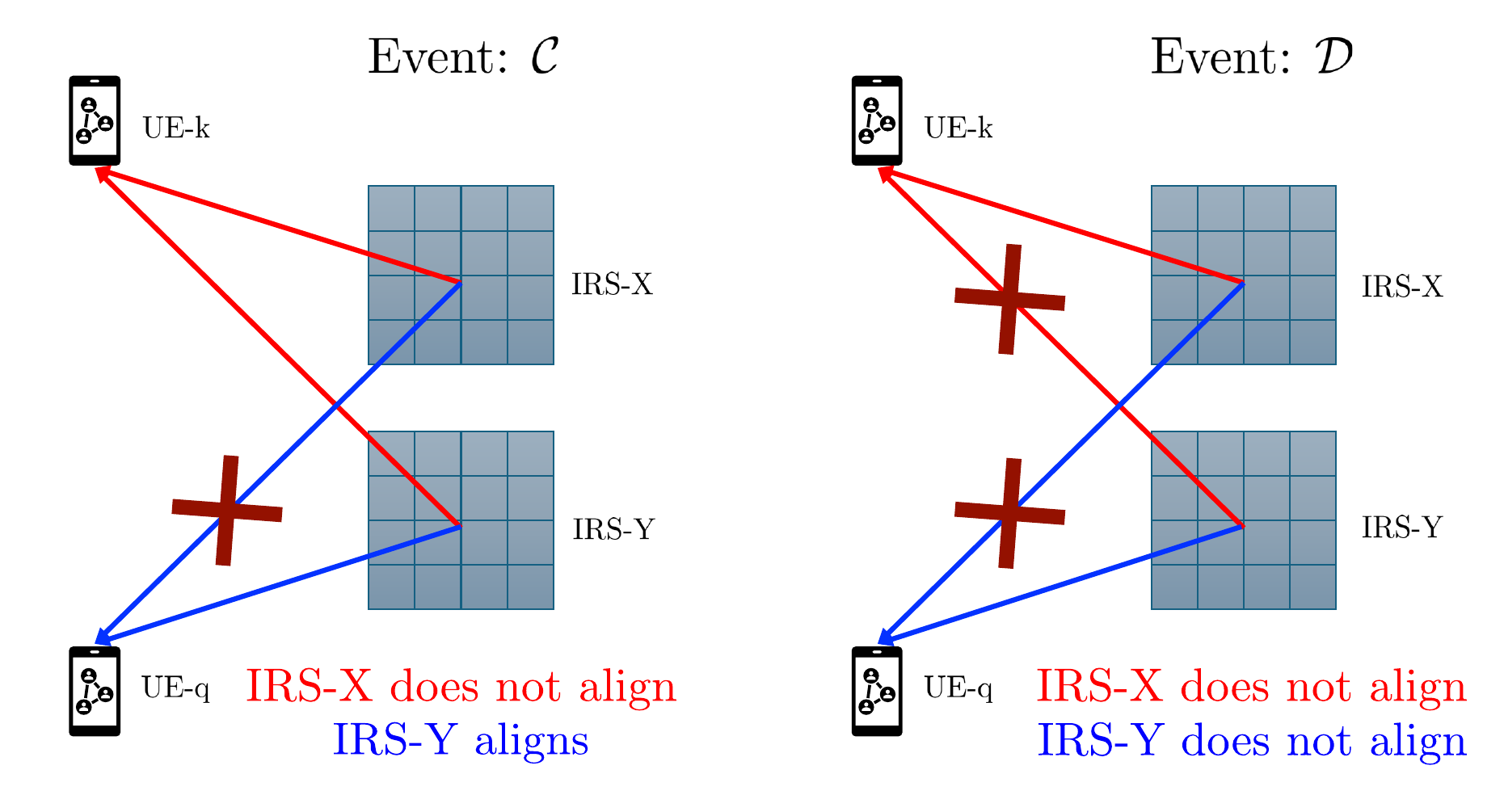}
	 	\caption{Illustration of all possible events in $2$-IRS aided $2$-MO system.}
 	\label{fig:System_Model_with-Events}
 	\vspace{-0.2cm}
 \end{figure}
\vspace{-0.15cm}
\subsection{Event $\mathcal{A}$: IRS-X and IRS-Y align to UE-$q$ and UE-$k$, resp.} \label{subsec: Event A}
    In this event, both IRS-X and IRS-Y align with one of the angles of the $L_{q1}$ and $L_{k2}$ OOB paths at UEs-$q$ and $k$, respectively.
     Now, since the alignment of IRS-X with UE-$q$'s channel is independent of the alignment of IRS-Y with UE-$k$'s channel, the probability of event $\mathcal{A}$ is 
     \vspace{-0.1cm}
    \begin{equation}\label{prob_A}
\Pr(\mathcal A)=    ({ L_{q1}}/{N_1})\times({L_{k2}}/{N_2}).
 \vspace{-0.1cm}
    \end{equation} 
       Now, under event $\mathcal{A}$, there exists indices $l_k^*$ and $l_q^*$ such that   
    \begin{align}\label{array_vector_config_A}
        l_k^* &= \arg_l \left\{N_2\mathbf{\dot{a}}_{N_2}^H(\omega^{(2)}_{X,k,l})\mathbf{\dot{a}}_{N_2}(\omega^{(2)}_{Y,q,1})=1\right\}, \ \text{and}\nonumber \\
        l_q^* &= \arg_l \left\{N_1\mathbf{\dot{a}}_{N_1}^H(\omega^{(1)}_{Y,q,l})\mathbf{\dot{a}}_{N_1}(\omega^{(1)}_{X,k,1})=1\right\},
    \end{align}
   where $\arg_l\{\cdot\}$ returns the index $l$ for which the condition in the braces is satisfied. In other words, the angles of $l_k^*$th and $l_q^*$th OOB paths at UE-$k$ and UE-$q$ match with the angles pointed by the phase configurations at IRS-Y and IRS-X, respectively.
   
   Then, using the expressions for IRS vectors in~\eqref{eq:optimum_vector_k}, and \eqref{eq:optimum_vector_q}, we simplify the channels of UE-$k$, $q$ in~\eqref{kth_channel_eq_1} and~\eqref{qth_channel_eq_1} as
    \begin{align}\label{final_eq_kth_user_Event_A}
        h_k \! = \!\frac{N_1}{\sqrt{L_{k1}}}\!\left|\gamma^{(1)}_{1,X}\gamma^{(1)}_{1,k}\right|\!\times\! e^{j \phi_1} \!+ \!\frac{N_2}{\sqrt{L_{k2}}}\!\left|\gamma^{(2)}_{l_k^*,X}\gamma^{(2)}_{l_k^*,k}\right|\!\times\! e^{j (\phi_2 + \phi_a)},
    \end{align}
   \vspace{-0.4cm}
    \begin{align}\label{final_eq_qth_user_Event_A}
        h_q \! = \!\frac{N_1}{\sqrt{L_{q1}}}\!\left|\gamma^{(1)}_{l_q^*,Y}\gamma^{(1)}_{l_q^*,q}\right|\!\times\! e^{j (\phi_1 + \phi_b)}\! +\! \frac{N_2}{\sqrt{L_{q2}}}\!\left|\gamma^{(2)}_{1,Y}\gamma^{(2)}_{1,q}\right|\!\times\! e^{j \phi_2},
    \end{align}
    respectively, where $\phi_a \triangleq \angle{\gamma^{(2)}_{l_k^*,X}} + \angle{\gamma^{(2)}_{l_k^*,k}} - \angle{\gamma^{(2)}_{1,Y}} - \angle{\gamma^{(2)}_{1,q}}$ and $\phi_b \triangleq \angle{\gamma^{(1)}_{l_q^*,Y}} + \angle{\gamma^{(1)}_{l_q^*,q}} - \angle{\gamma^{(1)}_{1,X}} - \angle{\gamma^{(1)}_{1,k}}$, denote the phase differences between OOB and in-band paths at IRSs X and Y, respectively. Now, if both the IRSs have to constructively add the received signals at both UEs-$k$, $q$, we need
    \vspace{-0.2cm}
    \begin{align}
        \phi_1 &= \phi_2 + \phi_a, \label{eq:inconsistent_eq_case_A_1} \\
       \text{ (and) } \phi_2 &= \phi_1 + \phi_b, \label{eq:inconsistent_eq_case_A_2}
    \end{align} 
    at IRS X and Y, respectively. However, since $\phi_a,\phi_b \in \mathcal{U}[-\pi,\pi) $ are i.i.d. random variables, \eqref{eq:inconsistent_eq_case_A_1} and \eqref{eq:inconsistent_eq_case_A_2} hold simultaneously with zero probability. That is, almost surely, both IRSs cannot be optimal for both UEs at the same time. 
        With this in mind, we analyze the $3$ schemes in Sec.~\ref{sec:Prob_statement}.
    \subsubsection{Joint-optimization of IRSs with MO cooperation} Here, the MOs jointly optimize the overall phase shifts $\phi_1$ and $\phi_2$ at the IRSs to maximize the weighted sum-SE of the UEs scheduled by both MOs. We first rewrite~\eqref{final_eq_kth_user_Event_A} and~\eqref{final_eq_qth_user_Event_A} as
   \vspace{-0.05cm}
    \begin{align}
        h_k & = \alpha  e^{j\phi_1} + \gamma  e^{j(\phi_2 + \phi_a)}, \label{eq_event_A_joint_optim_UE_k} \\  
        h_q & = \beta e^{j(\phi_1+\phi_b)}+\delta e^{j\phi_2},  \label{eq_event_A_joint_optim_UE_q}
    \end{align}
    where, $\alpha\triangleq\frac{N_1}{\sqrt{L_{k1}}}|\gamma^{(1)}_{1,X}\gamma^{(1)}_{1,k}|$, $\gamma\triangleq\frac{N_2}{\sqrt{L_{k2}}}|\gamma^{(2)}_{l_k^*,X}\gamma^{(2)}_{l_k^*,k}|$, $\beta\triangleq\frac{N_1}{\sqrt{L_{q1}}}|\gamma^{(1)}_{l_q^*,Y}\gamma^{(1)}_{l_q^*,q}|$, and $\delta\triangleq\frac{N_2}{\sqrt{L_{q2}}}|\gamma^{(2)}_{1,Y}\gamma^{(2)}_{1,q}|$. \textcolor{black}{Then, $\phi_1$ and $\phi_2$ are determined as $\phi_1^{\mathrm{opt}},\phi_2^{\mathrm{opt}} =$
	      \vspace{-0.2cm}
	    \begin{equation*}
	    \argmax_{\phi_1,\phi_2} \ \ w_k\log_2\left(1+\frac{P}{\sigma^2}|h_k|^2\right) + w_q\log_2\left(1+\frac{P}{\sigma^2}|h_q|^2\right)
		\end{equation*}
		where 
		$\phi_1$, $\phi_2$ are the overall phase shifts set by the BS-X and Y at IRS-X and Y, respectively, $w_k$ and $w_q$ are the weights allotted to the SE achieved by UEs $k$ and $q$, respectively.
				Let $\phi \triangleq \phi_2-\phi_1$, $x\triangleq1+\frac{P}{\sigma^2}(\alpha^2+\gamma^2)$, $v\triangleq2\alpha\gamma$, $y\triangleq1+\frac{P}{\sigma^2}(\beta^2+\delta^2)$ and $z\triangleq 2\beta\delta$. The above problem is equivalent to $\phi^{\mathrm{opt}} \triangleq$
	    	\begin{align*}
		\argmax_{\phi}\! f(\phi) \!=\! \left((x\!+\!v\cos(\phi\!+\!\phi_a))^{w_k}(y\!+\!z\cos(\phi\!-\!\phi_b))^{w_q}\right). 
	    \end{align*}
Since $f(\phi)$ depends only on the difference $ \phi = \phi_2-\phi_1$, the optimization variables can be reduced to a single variable $\phi$. Notably, the solution to this optimization problem inherently accounts for the operating SNR.} By the first order condition, $f'(\phi)=0$, which is
    \begin{equation}\label{eq:FONC_weighted}
		f(\phi)\left\{\frac{w_kv\sin(\phi+\phi_a)}{x+v\cos(\phi+\phi_a)}+\frac{w_qz\sin(\phi-\phi_b)}{y+z\cos(\phi-\phi_b)}\right\}=\!0.
	\end{equation}
\begin{figure*}[t]
    	\vspace{-0.7cm}
    	\begin{align}
    		f''(\phi)= \frac{(f'(\phi))^2}{f(\phi)}- f(\phi)\biggl\{\frac{w_kv\{x\cos(\phi+\phi_a)+v\}}{(x+v\cos(\phi+\phi_a))^2}+\frac{w_qz\{y\cos(\phi-\phi_b)+z\}}{(y+z\cos(\phi-\phi_b))^2}\biggl\}.
    		\label{eq:SOSC_weighted}
    	\end{align}
   	\hrule
    	\vspace{-0.5cm}
    \end{figure*}
    Since the roots of~\eqref{eq:FONC_weighted} do not admit a closed-form solution, we employ a low complexity Newton-Raphson's algorithm~\cite{chong2013introduction} to solve for $\phi$, which is outlined in Algorithm~\ref{alg:One_iteration_method}.  Although the Newton-Raphson method entails multiple iterations, we use only a single iteration to reduce complexity. In Sec.~\ref{sec:numerical_sections}, we numerically show that a single iteration with appropriate initialization yields comparable solutions to high-complexity off-the-shelf optimizers. In particular, we initialize $\phi$ based on the weights allotted to the MOs: we compute the weighted sum rate with $\phi= -\phi_a$ and $\phi = \phi_b$, and choose the value that yields the higher weighted sum-SE. Finally, with $\phi^{\mathrm{opt}}$ in hand, the ergodic SEs of UEs $k$ and $q$ are given by
        \begin{multline}
        \vspace{-0.3cm}
   \langle R_k | \mathcal{A} \rangle \approx \log_2\left(1+\frac{P}{\sigma^2}\left|({N_1}/{\sqrt{L_{k1}}})\left|\gamma^{(1)}_{1,X}\gamma^{(1)}_{1,k}\right|  e^{j\phi_1^{\mathrm{opt}}} \right. \right.  \\
   \left. \left. \hspace{1.2cm}+ ({N_2}/{\sqrt{L_{k2}}})\left|\gamma^{(2)}_{l_k^*,X}\gamma^{(2)}_{l_k^*,k}\right| e^{j(\phi_2^{\mathrm{opt}} + \phi_a)}\right|^2\right), \label{eq_Joint_optim_ergodic_SE_MO_X}
\vspace{-0.3cm}
	   \end{multline}
	   \vspace{-0.3cm}
   \begin{multline}
   \vspace{-0.3cm}
	 \langle R_q | \mathcal{A}\rangle \approx \log_2\left(1+ \frac{P}{\sigma^2}\left|({N_1}/{\sqrt{L_{q1}}})|\gamma^{(1)}_{l_q^*,Y}\gamma^{(1)}_{l_q^*,q}|  \right. \right.  \\
   \left. \left. \times e^{j(\phi_1^{\mathrm{opt}}+\phi_b)} + ({N_2}/{\sqrt{L_{q2}}})|\gamma^{(2)}_{1,Y}\gamma^{(2)}_{1,q}|e^{j\phi_2^{\mathrm{opt}}}\right|^2\right). \label{eq_Joint_optim_ergodic_SE_MO_Y}
   \vspace{-0.4cm}
    \end{multline}    
   \begin{figure}[t]
\vspace{-0.45cm}
    \begin{algorithm}[H] 
		\caption{Newton-Raphson based single iteration method}
		\label{alg:One_iteration_method}
		\begin{algorithmic}[1]
			\If{$w_k R_k(\phi_a,0)+w_q R_q(\phi_a,0) \geq w_k R_k(0,\phi_b)+ w_q R_q(0,\phi_b)$}
			\State $\phi_{init} \leftarrow -\phi_a$,
			\Else
			\State $\phi_{init} \leftarrow \phi_b$.
			\EndIf
			\State Obtain $f'(\phi)\leftarrow	-f(\phi)\left\{\frac{w_kv\sin(\phi+\phi_a)}{x+v\cos(\phi+\phi_a)}+\frac{w_qz\sin(\phi-\phi_b)}{y+z\cos(\phi-\phi_b)}\right\}$.
			\State Compute the second-derivative, $f''(\phi)$ using \eqref{eq:SOSC_weighted} on the top of this page.
			\State Update $\!\phi^{\mathrm{opt}}\leftarrow\phi_{init}-\frac{f'(\phi)}{f''(\phi)}\Bigl|_{\phi=\phi_{init}}$.
		\end{algorithmic}
	\end{algorithm} 
	\vspace{-1cm}
	\end{figure}
\subsubsection{Optimization of IRSs with time sharing} \label{sec:implementation_event_A}
Here, the MOs optimize the overall phase shifts of the IRSs to the UE scheduled by either MO-X or MO-Y, in a time-shared manner. Now, in the time slots used to optimize $\phi_1$ and $\phi_2$ to UE-$k$ scheduled by BS-X, from~\eqref{final_eq_kth_user_Event_A}, we need to choose $\phi_1=\phi_a$ and $\phi_2=0$, respectively. 
     Then,~\eqref{final_eq_kth_user_Event_A} and~\eqref{final_eq_qth_user_Event_A} simplify to
    \begin{align}
       \!\!\! h_k & = \left(\frac{N_1}{\sqrt{L_{k1}}}\left|\gamma^{(1)}_{1,X}\gamma^{(1)}_{1,k}\right| + \frac{N_2}{\sqrt{L_{k2}}}\left|\gamma^{(2)}_{l_k^*,X}\gamma^{(2)}_{l_k^*,k}\right|\right) e^{j\phi_a}, \label{final_eq_kth_user_event_A}\\
        \!\!\! h_q & = \frac{N_1}{\sqrt{L_{q1}}}\left|\gamma^{(1)}_{l_q^*,Y}\gamma^{(1)}_{l_q^*,q}\right| e^{j (\phi_a + \phi_b)} + \frac{N_2}{\sqrt{L_{q2}}}\left|\gamma^{(2)}_{1,Y}\gamma^{(2)}_{1,q}\right|. \label{final_eq_qth_user_event_A}
    \end{align}
  We have the following lemma to characterize the in-band SE.
    \begin{lemma}\label{lemma:lemma1}
        Let $\{X_i\}_{i=1}^{N}$ be i.i.d random variables such that  $X_i \sim \mathcal{CN}(0,1)$.  If $M\triangleq \max(|X_1|,|X_2|,\ldots,|X_N|)$, and $G \triangleq \max(|X_1|^2\!,\!|X_2|^2,\ldots,\!|X_N|^2)$, the expected values of $M,G$ are
         \vspace{-0.2cm}
                \begin{equation}
                \vspace{-0.2cm}
        \mathbb{E}[M] = f(N) \triangleq N\sum\nolimits_{n=0}^{N-1} {{N-1}\choose n} {(-1)}^n \frac{1}{(n+1)^{\frac{3}{2}}}\sqrt{\frac{\pi}{4}},
        \end{equation} 
        and 
        \vspace{-0.05cm}
        \begin{equation}\label{eq:EG}
        \mathbb{E}[G] = g(N) \triangleq N\sum\nolimits_{n=0}^{N-1} {{N-1}\choose n} {(-1)}^n \frac{1}{(n+1)^{2}},
        	\vspace{-0.05cm}
        \end{equation} respectively. 
    \end{lemma}
    	\vspace{-0.3cm}
    \begin{proof}
        Straightforward and omitted due to lack of space. 
    \end{proof}
    \vspace{-0.3cm}
    Using Lemma~\ref{lemma:lemma1}, we can show that $\mathbb E \left[|\gamma^{(1)}_{1,X}\gamma^{(1)}_{1,k}|\right] = (f(L_{k1}))^2 \sqrt{\beta^{(1)}_{X,k}} $, and $\mathbb E \left[|\gamma^{(1)}_{1,X}\gamma^{(1)}_{1,k}|^2\right] = (g(L_{k1}))^2 {\beta^{(1)}_{X,k}}$, where $\beta^{(1)}_{X,k} \triangleq \beta^{(1)}_{X}\beta^{(1)}_{k}$ and $\beta^{(2)}_{X,k} \triangleq \beta^{(2)}_{X}\beta^{(2)}_{k}$. Similarly, let $\beta^{(1)}_{Y,q} \triangleq \beta^{(1)}_{Y}\beta^{(1)}_{q}$ and $\beta^{(2)}_{Y,q}\triangleq \beta^{(2)}_{Y}\beta^{(2)}_{q}$. 
        Conditioned on event $\mathcal{A}$, by Jensen's approximation, the ergodic SE of UE-$k$, $q$ is
    \begin{align}\label{jansen_rate}
        \langle R_i|\mathcal{A}\rangle 
        \!&\approx\log_2\left(1+\mathbb{E}[|h_i|^2]P/\sigma^2\right), \ \ i\in \{k,q\} , 
    \end{align}
    where using~\eqref{final_eq_kth_user_event_A}, \eqref{final_eq_qth_user_event_A}, $\mathbb{E}\left[|h_k|^2\right]$, $\mathbb{E}[|h_q|^2]$ are derived below:
    \begin{multline}\label{avg_channel_kth}
        \mathbb{E}[|h_k|^2] 
        = ({N_1^2}/{L_{k1}})(g(L_{k1}))^2\beta^{(1)}_{X,k}+ ({N_2^2}/{L_{k2}})\beta^{(2)}_{X,k} \\+ ({\pi N_1N_2} /{2\sqrt{L_{k1}L_{k2}}}) (f(L_{k1}))^2 \sqrt{\beta^{(1)}_{X,k}\beta^{(2)}_{X,k}},
    \end{multline}
    \vspace{-0.3cm}
    \begin{multline}
        \mathbb{E}\left[|h_q|^2\right] 
        = ({N_1^2}/{L_{q1}})\beta^{(1)}_{Y,q} + ({N_2^2}/{L_{q2}})\beta^{(2)}_{Y,q} (g(L_{q2}))^2 \\
        \hspace{-0.1cm}+\frac{\pi N_1N_2} {2\sqrt{L_{q1}L_{q2}}}(f(L_{q2}))^2 \sqrt{\beta^{(1)}_{Y,q}\beta^{(2)}_{Y,q}}\times\mathbb{E}[\cos(\phi_a+\phi_b)].\!\!\!\!\!\!
    \end{multline}
    Since $\phi_a,\phi_b \sim \mathcal{U}[-\pi,\pi)$, we have $\mathbb{E}[\cos(\phi_a+\phi_b)]=0$. So, 
    \begin{equation}\label{avg_channel_qth}
    	\vspace{-0.05cm}
   \mathbb{E}\left[|h_q|^2\right] = ({N_1^2}/{L_{q1}})\beta^{(1)}_{Y,q} + ({N_2^2}/{L_{q2}})(g(L_{q2}))^2\beta^{(2)}_{Y,q}.
    \end{equation}
    Using~\eqref{avg_channel_kth} and~\eqref{avg_channel_qth} in~\eqref{jansen_rate}, the ergodic SEs of UE-$k$, $q$ when $\phi_1$, $\phi_2$ are optimized only for UE-$k$ can be obtained as
    \begin{align}
        \langle R_k|\mathcal{A} \rangle &\approx \log_2\biggl(1+\frac{P}{\sigma^2}\biggl\{\frac{N_1^2}{L_{k1}}(g(L_{k1}))^2\beta^{(1)}_{X,k}+\frac{N_2^2}{L_{k2}}\beta^{(2)}_{X,k}  \nonumber \\
        &\hspace{-0.6cm}+ ({\pi N_1N_2}/{2\sqrt{L_{k1}L_{k2}}}) \!(f(L_{k1}))^2\sqrt{\beta^{(1)}_{X,k}\beta^{(2)}_{X,k}}\biggl\}\biggl), \label{Rate_case_1_EventA} \\
        \langle R_q|\mathcal{A} \rangle &\approx \log_2\left(\!1\!+\frac{P}{\sigma^2}\left\{\!\frac{N_1^2}{L_{q1}}\beta^{(1)}_{Y,q} + \frac{N_2^2}{L_{q2}}\!(g(L_{q2}))^2\beta^{(2)}_{Y,q}\right\}\!\right)\!. \label{eq:rate_UE_q_ignoring_OOB_IRS}
    \end{align}
   Similarly, in the time slots used to optimize $\phi_1$, $\phi_2$ for UE-$q$, from~\eqref{final_eq_qth_user_Event_A}, $\phi_1=0$ and $\phi_2=\phi_b$. Then, 
       similar to~\eqref{Rate_case_1_EventA} and~\eqref{eq:rate_UE_q_ignoring_OOB_IRS}, the ergodic SE of UE-$k$, $q$ is    
        \begin{align}
       \!\!\! \langle R_k|\mathcal{A} \rangle &\approx \log_2\biggl(1+\frac{P}{\sigma^2}\biggl\{\frac{N_1^2}{L_{k1}}(g(L_{k1}))^2\beta^{(1)}_{X,k}+\frac{N_2^2}{L_{k2}}\beta^{(2)}_{X,k}\biggl\}\biggl),  \label{eq:rate_UE_k_ignoring_OOB_IRS} \\
       \!\!\! \langle R_q|\mathcal{A} \rangle &\approx \log_2\biggl(1+\frac{P}{\sigma^2}\biggl\{\frac{N_1^2}{L_{q1}}\beta^{(1)}_{Y,q} + \frac{N_2^2}{L_{q2}}(g(L_{q2}))^2\beta^{(2)}_{Y,q} \nonumber \\
        &\hspace{-0.3cm}+ ({\pi N_1N_2}/ {2\sqrt{L_{q1}L_{q2}}})(f(L_{q2}))^2 \sqrt{\beta^{(1)}_{Y,q}\beta^{(2)}_{Y,q}}\biggl\}\biggl). \label{Rate_case_2EventA_UE_q}
    \end{align}
    Therefore, using the expressions in~\eqref{Rate_case_1_EventA} to \eqref{Rate_case_2EventA_UE_q}, under a time-sharing scheme, with a fraction of time-slots, say  $\zeta T_c$, $\zeta \in (0,1)$ ($T_c$ is the coherence time) used to optimize overall phase shifts of both IRSs to UE-$k$, and remaining $(1-\zeta) T_c$ slots used to tune the overall phase shifts at both IRSs to serve UE-$q$, the overall ergodic SEs of the MOs are as given in~\eqref{eq:time_share_UE_k},~\eqref{eq:time_share_UE_q}, respectively, at the bottom of this page.
 	\begin{figure*}[b]
 		\vspace{-0.4cm}
 		\hrule
 		\begin{align}
 			\langle R_k|\mathcal{A} \rangle\textcolor{black}{_{\text{time-sharing}}} &\approx (1-\zeta)\log_2\left(1+\frac{P}{\sigma^2}\left\{({N_1^2}/{L_{k1}})(g(L_{k1}))^2\!\beta^{(1)}_{X,k}+ ({N_2^2}/{L_{k2}}) \beta^{(2)}_{X,k}\right\}\right) \nonumber \\ &\hspace{1.5cm}+\zeta\log_2\biggl(1\!+\!\frac{P}{\sigma^2}\biggl\{\frac{N_1^2}{L_{k1}}(g(L_{k1}))^2\beta^{(1)}_{X,k} \!+\! \frac{N_2^2}{L_{k2}}\beta^{(2)}_{X,k}+\frac{\pi N_1N_2} {2\sqrt{L_{k1}L_{k2}}}(f(L_{k1}))^2 \sqrt{\beta^{(1)}_{X,k}\beta^{(2)}_{X,k}}\biggl\}\biggl), \label{eq:time_share_UE_k} \\
 			\langle R_q|\mathcal{A} \rangle\textcolor{black}{_{\text{time-sharing}}} &\approx \zeta\log_2\left(1+\frac{P}{\sigma^2}\left\{({N_1^2}/{L_{q1}})\beta^{(1)}_{Y,q} + ({N_2^2}/{L_{q2}})(g(L_{q2}))^2\beta^{(2)}_{Y,q}\right\}\right)  \nonumber\\
 			\vspace{-0.2cm} 
 			&\hspace{0.8cm}+(1-\zeta)\log_2\biggl(1+\frac{P}{\sigma^2}\biggl\{\frac{N_1^2}{L_{q1}}\beta^{(1)}_{Y,q} + \frac{N_2^2}{L_{q2}}(g(L_{q2}))^2\beta^{(2)}_{Y,q}+\frac{\pi N_1N_2} {2\sqrt{L_{q1}L_{q2}}}(f(L_{q2}))^2 \sqrt{\beta^{(1)}_{Y,q}\beta^{(2)}_{Y,q}}\biggl\}\biggl). \label{eq:time_share_UE_q}
 		\end{align}
 		\vspace{-0.7cm}
 	\end{figure*}
 	
	\subsubsection{Optimization of IRSs without cooperation} In this case, the MOs optimize their IRSs by only considering the in-band channels at their UEs. So, the overall phase shifts can be set to $\phi_1=\phi_2=0$ as per~\eqref{eq_scheme_3_overall_phase}. Hence, the ergodic SE at UE-$k$, $q$ is given by~\eqref{eq:rate_UE_k_ignoring_OOB_IRS} and~\eqref{eq:rate_UE_q_ignoring_OOB_IRS}, respectively.
       Clearly, the cross terms (as in~\eqref{Rate_case_1_EventA} and~\eqref{Rate_case_2EventA_UE_q}) do not appear in these expressions as the signals from the IRSs do not add coherently at the UEs. 
       \subsection{Evt. $\mathcal{B}$: IRS-X aligns to UE-$q$, IRS-Y does not align to UE-$k$}
    Here, one of the $L_{q1}$ OOB paths at UE-$q$ aligns with IRS-X's beam, while none of the $L_{k2}$ OOB paths at UE-$k$ match with IRS-Y's beam. Hence, the probability of event $\mathcal{B}$ is  
    \begin{equation}\label{prob_B}
        \Pr(\mathcal{B})  =  ({ L_{q1}}/{N_1})\cdot\left(1-({L_{k2}}/{N_2})\right).
    \vspace{-0.1cm}
    \end{equation} 
Note that there is no need for joint optimization of $\phi_1,\phi_2$ here because only one of the MO's UE gets signals reflected from both IRSs. Then, the channels at UE-$k, q$ are 
    \begin{align}\label{final_eq_kth_user}
        &\hspace{-0.15 cm}h_k  = ({N_1}/{\sqrt{L_{k1}}})\left|\gamma^{(1)}_{1,X}\gamma^{(1)}_{1,k}\right|\times e^{j \phi_1}, \\
        &\hspace{-0.15 cm}h_q \! = \!\frac{N_1}{\sqrt{L_{q1}}}\!\left|\gamma^{(1)}_{l_q^*,Y}\gamma^{(1)}_{l_q^*,q}\right|\! e^{j (\phi_1 + \phi_b)}\! +\! \frac{N_2}{\sqrt{L_{q2}}}\!\left|\gamma^{(2)}_{1,Y}\gamma^{(2)}_{1,q}\right|\! e^{j \phi_2}.\!\! \label{event2_qth_equation}
    \end{align}
	Since IRS-Y does not contribute to UE-$k$, the SE achieved by UE-$k$ is only due to IRS-X and is independent of $\phi_1$. On the other hand, for coherent addition of in-band and OOB paths to maximize $|h_q|^2$ in \eqref{event2_qth_equation}, we need $\phi_2 = \phi_1 + \phi_b$, and we choose the simplest solution for this, namely $\{\phi_1=0;\phi_2=\phi_b\}$ so that BS-X need not apply any additional overall phase. Then,
    \begin{align}
        \hspace{-0.15 cm}h_k & = ({N_1}/{\sqrt{L_{k1}}})\left|\gamma^{(1)}_{1,X}\gamma^{(1)}_{1,k}\right|, \ \text{and}\nonumber \\
        \hspace{-0.15 cm}h_q & = \left(\frac{N_1}{\sqrt{L_{q1}}}\left|\gamma^{(1)}_{l_q^*,Y}\gamma^{(1)}_{l_q^*,q}\right| + \frac{N_2}{\sqrt{L_{q2}}}\left|\gamma^{(2)}_{1,Y}\gamma^{(2)}_{1,q}\right|\right) \!\times e^{j\phi_b}.
    \end{align}
So, the ergodic SE of UE-$k$, $q$ when $\phi_1$, $\phi_2$ are optimized to UE-$q$ in \textcolor{black}{$(1-\zeta)T_c$ time slots} can be obtained as
 \begin{equation}
        \langle R_k|\mathcal B \rangle \approx \log_2\biggl(1+\frac{P}{\sigma^2}\biggl\{\frac{N_1^2}{L_{k1}}(g(L_{k1}))^2\beta^{(1)}_{X,k}\biggl\}\biggl),  \label{Rate_EventB_k} 
    \end{equation} and~\eqref{eq:time_share_UE_q}, respectively.
Similarly, with no MO cooperation, the SE at UE-$k,q$ are as in~\eqref{Rate_EventB_k} and \eqref{eq:rate_UE_q_ignoring_OOB_IRS}, respectively. 

To summarize, \textcolor{black}{when event $\mathcal{B}$ occurs}, for 1) joint optimization with MO cooperation, the SEs of UEs $k$ and $q$ are given by~\eqref{Rate_EventB_k} and~\eqref{Rate_case_2EventA_UE_q}, respectively; for 2) optimization with time-sharing schemes, the SEs of UEs $k$ and $q$ are given by \eqref{Rate_EventB_k} and~\eqref{eq:time_share_UE_q}, respectively.
With 3) no cooperation, the SEs are given by~\eqref{Rate_EventB_k} and \eqref{eq:rate_UE_q_ignoring_OOB_IRS}, respectively.
\subsection{Evt. $\mathcal{C}$: IRS-X does not align to UE-$q$, IRS-Y aligns to UE-$k$}\label{sec:Case_C}
    This event is the complement of event $\mathcal{B}$ described above, i.e., one of the $L_{k2}$ OOB paths at UE-$k$ aligns with IRS-Y, but none of the $L_{q1}$ OOB paths at UE-$q$ matches with the beam formed by IRS-X. Hence, 
    \begin{equation}\label{prob_C}
        \Pr(\mathcal C) = \left(1-({L_{q1}}/{N_1})\right)\cdot({L_{k2}}/{N_2}).
    \end{equation}
   Following the analysis similar to event $\mathcal{B}$, the final expressions for the ergodic SEs of UE-$k$, $q$ when $\phi_1$, $\phi_2$ are optimized to UE-$k$ for $\zeta T_c$ time slots are given in~\eqref{eq:time_share_UE_k}, and as
        \begin{equation}\label{Rate_EventC_UE_q}
        \langle R_q|\mathcal C \rangle \approx \log_2\left(1+\frac{P}{\sigma^2}\left\{\frac{N_2^2}{L_{q2}}(g(L_{q2}))^2\beta^{(2)}_{Y,q} \right\}\right), 
    \end{equation} respectively.
   Similarly, without MO cooperation, the SEs at UEs $k$ and $q$ are given by \eqref{eq:rate_UE_k_ignoring_OOB_IRS} and~\eqref{Rate_EventC_UE_q}, respectively.
   
   Thus, whenever event $\mathcal{C}$ occurs, for 1) joint optimization with MO cooperation, the SEs of UEs $k$, $q$ are given by~\eqref{Rate_case_1_EventA} and~\eqref{Rate_EventC_UE_q}; for 2) optimization with time-sharing, the SEs of UEs $k$, $q$ are given by \eqref{eq:time_share_UE_k} and~\eqref{Rate_EventC_UE_q}, respectively.
With 3) no cooperation, the SEs are as in~\eqref{eq:rate_UE_k_ignoring_OOB_IRS}, and~\eqref{Rate_EventC_UE_q}, respectively.
    \vspace{-0.2cm}
\subsection{Event $\mathcal D$: IRS-X, Y do not align to UE-$q$, $k$, respectively}
    In this final event, none of the $L_{q1}$ and $L_{k2}$ OOB paths match with IRS-X and IRS-Y, respectively. The probability of this event is given by
    \begin{equation}\label{prob_D}
    	\vspace{-0.1cm}
        \Pr(\mathcal D) = \left(1-({L_{q1}}/{N_1})\right)\cdot\left(1-({L_{k2}}/{N_2})\right).
    \end{equation} Since none of the IRSs align with an OOB UE, this event completely obviates the need for tuning $\phi_1,\phi_2$. The channel coefficients under this event are 
    \begin{align}
      \!\!  h_k \! = \!\frac{N_1}{\sqrt{L_{k1}}}\!\left|\gamma^{(1)}_{1,X}\gamma^{(1)}_{1,k}\right|\! e^{j \phi_1}, \  
        h_q \! = \! \frac{N_2}{\sqrt{L_{q2}}}\!\left|\gamma^{(2)}_{1,Y}\gamma^{(2)}_{1,q}\right|\! e^{j \phi_2},
    \end{align}
    respectively, and the ergodic SEs of UE-$k$, $q$ are given by
    \begin{align}\label{Rate_EventD}
        \langle R_k|\mathcal D \rangle = \langle R_k|\mathcal B \rangle, \text{  and }  \langle R_q|\mathcal D \rangle = \langle R_q|\mathcal C \rangle,
    \end{align} where $\langle R_k|\mathcal B \rangle$ and $\langle R_q|\mathcal C \rangle$ are as per~\eqref{Rate_EventB_k},~\eqref{Rate_EventC_UE_q}, respectively.
    
     To summarize, whenever event $\mathcal{D}$ occurs, for 1) joint optimization with MO cooperation, 2) optimization with time-sharing, and 3) with no MO schemes cooperation, the SEs of UEs $k$ and $q$ are the same for all $3$ schemes as given in~\eqref{Rate_EventD}.
We now state our main result on the overall ergodic sum-SEs of both the MOs under round-robin (RR) scheduling.
    \begin{theorem}\label{thm:theorem-1}
    Under the SV channel model in the mmWave bands, when MOs X \& Y control an IRS each to serve their subscribed UEs, the ergodic sum-SE of MOs X and Y under RR scheduling are characterized in the following: 
    \begin{enumerate}[leftmargin=*]
			\item Optimization of IRSs with time sharing where $\zeta T_c$ time slots are alloted to MO-X and $(1-\zeta)T_c$ time slots are alloted to MO-Y: $\langle R_X^{\zeta} \rangle_{\textrm{TS}}$, $\langle R_Y^{\zeta} \rangle_{\textrm{TS}}$ as in~\eqref{final_rate_operator_X} and~\eqref{final_rate_operator_Y}, with $\mathsf{CO}=1$.
    	\item Joint-optimization of IRSs with MO cooperation: The ergodic rates for MO-X and MO-Y are given as in~\eqref{eq_Joint_optim_full_ergodic_SE_MO_X} and~\eqref{eq_Joint_optim_full_ergodic_SE_MO_Y}, respectively, where $\phi_1^{\mathrm{opt}}$ and $\phi_2^{\mathrm{opt}}$ are chosen such that $\phi^{\mathrm{opt}} = \phi_2^{\mathrm{opt}} - \phi_1^{\mathrm{opt}}$ is a solution obtained from Algorithm~\ref{alg:One_iteration_method}.		Further, an upper bound on the SEs $\langle R_X \rangle_{\textrm{JO}}$ and $\langle R_Y \rangle_{\textrm{JO}}$ without relying on  Algorithm~\ref{alg:One_iteration_method} can be obtained as $\langle R_X \rangle_{\textrm{JO}} \leq \langle R_X^{\zeta} \rangle_{\textrm{TS}}\Big\vert_{\zeta=1}$ and $\langle R_Y \rangle_{\textrm{JO}} \leq \langle R_Y^{\zeta} \rangle_{\textrm{TS}}\Big\vert_{\zeta=0}$, respectively, 
	where $\langle R_X^{\zeta} \rangle_{\textrm{TS}}$, $\langle R_Y^{\zeta} \rangle_{\textrm{TS}}$ are given as in~\eqref{final_rate_operator_X} and~\eqref{final_rate_operator_Y}, with $\mathsf{CO}=1$.	    	
	\item Optimization of IRSs without MO cooperation: $\langle R_X \rangle_{\textrm{NCO}}$, $\langle R_Y \rangle_{\textrm{NCO}}$ as given in~\eqref{final_rate_operator_X} and \eqref{final_rate_operator_Y} with $\mathsf{CO}=0$.
    \end{enumerate}
    \begin{figure*}[t]
    	\vspace{-0.8cm}
    \begin{align}
    	\hspace{-0.5cm}\langle R_X^{\zeta} \rangle_{\textrm{TS}} \!&\approx \! \frac{1}{K}\sum_{k=1}^K\! \left\{\left(\!1-\!\frac{ L_{k2}}{N_2}\right)\!\log_2\!\left(\!1\!+\!\frac{P}{\sigma^2}\frac{N_1^2}{L_{k1}}(g(L_{k1}))^2\beta^{(1)}_{X,k}\right) \!+\! (1-\zeta)\frac{L_{k2}}{N_2}\log_2\!\left(\!1\!+\!\frac{P}{\sigma^2}\left\{\frac{N_1^2}{L_{k1}}(g(L_{k1}))^2\!\beta^{(1)}_{X,k}\!+\! \frac{N_2^2}{L_{k2}}\beta^{(2)}_{X,k}\right\}\right)\right. \nonumber\\ 
	&\hspace{0.8cm}\left.+\zeta\frac{ L_{k2}}{N_2}\log_2\left(1+\frac{P}{\sigma^2}\left\{\frac{N_1^2}{L_{k1}}(g(L_{k1}))^2\beta^{(1)}_{X,k}+\frac{N_2^2}{L_{k2}}\beta^{(2)}_{X,k} 
    	 + \textcolor{black}{\mathbbm{1}_{\{\mathsf{CO} = 1\}}} \frac{\pi N_1N_2(f(L_{k1}))^2\sqrt{\beta^{(1)}_{X,k}\beta^{(2)}_{X,k}}} {2\sqrt{L_{k1}L_{k2}}} \right\}\right)\right\},  \label{final_rate_operator_X} \\
    	\!\!\langle R_Y^{\zeta} \rangle_{\textrm{TS}} &\approx \frac{1}{Q} \sum_{q=1}^Q \left\{\left(1-\frac{ L_{q1}}{N_1}\right)\log_2\left(1+\frac{P}{\sigma^2}\frac{N_2^2}{L_{q2}}(g(L_{q2}))^2\beta^{(2)}_{Y,q}\right) + \zeta\frac{L_{q1}}{N_1}\log_2              
    	\left(1+\frac{P}{\sigma^2}\left\{\frac{N_1^2}{L_{q1}}\beta^{(1)}_{Y,q}+\frac{N_2^2}{L_{q2}}(g(L_{q2}))^2\beta^{(2)}_{Y,q} \right\}\right)\right.  \nonumber \\
    	&\hspace{0.2cm}\left.+(1-\zeta)\frac{ L_{q1}}{N_1}\log_2              
    	\left(1+\frac{P}{\sigma^2}\left\{\frac{N_1^2}{L_{q1}}\beta^{(1)}_{Y,q}+\frac{N_2^2}{L_{q2}}(g(L_{q2}))^2\beta^{(2)}_{Y,q} + \textcolor{black}{\mathbbm{1}_{\{\mathsf{CO} = 1\}}} \frac{\pi N_1N_2(f(L_{q2}))^2\sqrt{\beta^{(1)}_{Y,q}\beta^{(2)}_{Y,q}}} {2\sqrt{L_{q1}\!L_{q2}}} \right\}\right)\right\}. \label{final_rate_operator_Y} 
          \end{align}
          \hrule
          \vspace{-0.4cm}
          \end{figure*}
    \begin{figure*}
    \begin{align}
   & \langle R_X \rangle_{\textrm{JO}} \approx \dfrac{1}{K}\sum_{k=1}^K\left\{\frac{L_{q1}L_{k2}}{N_1N_2}\log_2\left(1+\frac{P}{\sigma^2}\left|\frac{N_1}{\sqrt{L_{k1}}}\left|\gamma^{(1)}_{1,X}\gamma^{(1)}_{1,k}\right|  e^{j\phi_1^{\mathrm{opt}}} + \frac{N_2}{\sqrt{L_{k2}}}\left|\gamma^{(2)}_{l_k^*,X}\gamma^{(2)}_{l_k^*,k}\right| e^{j(\phi_2^{\mathrm{opt}} + \phi_a)}\right|^2\right) + \left(1-\frac{L_{k2}}{N_2}\right)\right. \nonumber\\
   & \hspace{1.8cm}\times \log_2\biggl(1+\frac{P}{\sigma^2}\biggl\{\frac{N_1^2}{L_{k1}}(g(L_{k1}))^2\beta^{(1)}_{X,k}\biggl\}\biggl) + \left(1-\frac{L_{q1}}{N_1}\right)\frac{L_{k2}}{N_2}\log_2\left(1+\frac{P}{\sigma^2}\left\{\frac{N_1^2}{L_{k1}}(g(L_{k1}))^2\beta^{(1)}_{X,k} \right. \right. \nonumber \\
    	 & \hspace{6cm}\left. \left. \left. +({N_2^2}/{L_{k2}})\beta^{(2)}_{X,k} + \left({\pi N_1N_2(f(L_{k1}))^2\sqrt{\beta^{(1)}_{X,k}\beta^{(2)}_{X,k}}}\Big/ {2\sqrt{L_{k1}L_{k2}}}\right) \right\}\right)\right\}, \label{eq_Joint_optim_full_ergodic_SE_MO_X}\\
	& \langle R_Y \rangle_{\textrm{JO}} \approx \dfrac{1}{Q}\sum_{q=1}^Q \left\{ \frac{L_{q1}L_{k2}}{N_1N_2}\log_2\left(1+ \frac{P}{\sigma^2}\left|\frac{N_1}{\sqrt{L_{q1}}}|\gamma^{(1)}_{l_q^*,Y}\gamma^{(1)}_{l_q^*,q}|e^{j(\phi_1^{\mathrm{opt}}+\phi_b)}  + \frac{N_2}{\sqrt{L_{q2}}}|\gamma^{(2)}_{1,Y}\gamma^{(2)}_{1,q}|e^{j\phi_2^{\mathrm{opt}}}\right|^2\right) + \left(1-\frac{L_{q1}}{N_1}\right) \nonumber \right.\\
	&\hspace{1.8cm}\times \log_2\left(1+\frac{P}{\sigma^2}\left\{\frac{N_2^2}{L_{q2}}(g(L_{q2}))^2\beta^{(2)}_{Y,q} \right\}\right) + \left(1-\frac{L_{k2}}{N_2}\right) \frac{ L_{q1}}{N_1} \log_2              
    	\left(1+\frac{P}{\sigma^2}\left\{ \frac{N_2^2}{L_{q2}}(g(L_{q2}))^2\beta^{(2)}_{Y,q} \right. \right. \nonumber \\
	&\hspace{6cm}\left. \left. \left. + ({N_1^2}/{L_{q1}})\beta^{(1)}_{Y,q} + \left({\pi N_1N_2(f(L_{q2}))^2\sqrt{\beta^{(1)}_{Y,q}\beta^{(2)}_{Y,q}}}\Big/ {2\sqrt{L_{q1}\!L_{q2}}}\right) \right\}\right)\right\}. \label{eq_Joint_optim_full_ergodic_SE_MO_Y}
    \end{align}
    \hrule
    \vspace{-0.5cm}
    \end{figure*}
    \begin{proof}
        We only prove~\eqref{final_rate_operator_X} (and~\eqref{eq_Joint_optim_full_ergodic_SE_MO_X}), for MO-X; the proof of \eqref{final_rate_operator_Y} (and~\eqref{eq_Joint_optim_full_ergodic_SE_MO_Y}) is similar. Using the law of total expectation, at a given UE-$k$, the ergodic SE (for all three schemes) is
        \vspace{-0.1cm}
        \begin{equation}\label{eq:total_law_expectation}
        	\vspace{-0.1cm}
        \langle R_k \rangle = \!\sum\nolimits_{i \in \left\{\mathcal{A},\mathcal{B},\mathcal{C},\mathcal{D}\right\}}\!\langle R_k|i \rangle\Pr(i), 
        \end{equation} 
        where the probabilities can be found using \eqref{prob_A}, \eqref{prob_B}, \eqref{prob_C}, and \eqref{prob_D} for events $\mathcal{A},\mathcal{B},\mathcal{C}$, and $\mathcal{D}$, respectively. As a consequence, under RR scheduling, the ergodic sum-SE of MO-X is
        \begin{equation}\label{eq_RR_rate_UE_k}
        \langle R_X \rangle= ({1}/{K}) \sum\nolimits_{k=1}^K \langle R_k \rangle.
        \end{equation}
        Next, the scheme-specific ergodic SE is characterized below. 
        \begin{itemize}[leftmargin=*]
        \item Optimization of IRSs with time-sharing: In this scheme, $\langle R_k|i \rangle$ can be obtained from~\eqref{eq:time_share_UE_k}, \eqref{Rate_EventB_k}, \eqref{eq:time_share_UE_k}, and~\eqref{Rate_EventD} for events $\mathcal{A},\mathcal{B},\mathcal{C}$, and $\mathcal{D}$, respectively, then using~\eqref{eq:total_law_expectation} and \eqref{eq_RR_rate_UE_k}, the result in~\eqref{final_rate_operator_X} for MO-X follows.
        \item Joint-optimization of IRSs with MO cooperation: We obtain the values of $\langle R_k|i \rangle$ by using~\eqref{eq_Joint_optim_ergodic_SE_MO_X}, \eqref{Rate_EventB_k}, \eqref{Rate_case_1_EventA}, and~\eqref{Rate_EventD} for events $\mathcal{A}, \mathcal{B}, \mathcal{C}$, and $\mathcal{D}$, respectively. 
        Then, substituting these values into~\eqref{eq:total_law_expectation} and in~\eqref{eq_RR_rate_UE_k} yields the expression in~\eqref{eq_Joint_optim_full_ergodic_SE_MO_X}. \\
    \vspace{-0.4cm}
    
      \hspace{0.2cm} Further, the ergodic SE obtained at a UE by jointly optimizing the IRSs is, at most, the SE obtained by optimizing the overall phases of both the IRSs to that UE in all time slots when it is scheduled. Thus, the SE with joint optimization is upper bounded by the expressions for the time-sharing scheme given above, but with $\zeta=1$ and $\zeta=0$ for MO-X and MO-Y, respectively. This establishes the upper bounds in the Theorem.
        \item For the no cooperation scheme, $\langle R_k|i \rangle$ can be found from~\eqref{eq:rate_UE_k_ignoring_OOB_IRS}, \eqref{Rate_EventB_k}, \eqref{eq:rate_UE_k_ignoring_OOB_IRS}, \eqref{Rate_EventB_k} for events $\mathcal{A},\mathcal{B},\mathcal{C}$, and $\mathcal{D}$, respectively, and using them in~\eqref{eq:total_law_expectation}, \eqref{eq_RR_rate_UE_k}, the result follows.
\end{itemize}
This completes the proof of the theorem.
    \end{proof}
\end{theorem}
We interpret Theorem~\ref{thm:theorem-1} as follows. At MO-X, in the time-sharing scheme, among the time slots in which the OOB IRS-Y aligns with UE-$k$ (which happens with probability $L_{k2}/N_2$), for a $\zeta$ fraction of the total time slots, the overall phase shifts of both IRSs are optimized to UE-$k$ only and procures an array gain that scales as $\mathcal{O}\left((N_1 + N_2)^2\right)$. For the other $(1-\zeta)$ fraction of time slots, since the IRSs add the signals coherently at UE-$q$ served by MO-Y, the array gain at UE-$k$ is only due to an incoherent addition of signals from the IRSs, i.e., it scales as $\mathcal{O}(N_1^2 + N_2^2)$. On the other hand, when IRS-Y does not align with UE-$k$ (which happens with probability $1-L_{k2}/N_2$), the array gain scales only as $\mathcal{O}\left(N_1^2\right)$. Similarly, under the joint-optimization scheme, whenever IRS-Y aligns with UE-$k$, but IRS-X does not align with UE-$q$, i.e., with event $\mathcal{C}$ which happens with probability $(1-L_{q1}/N_2)L_{k2}/N_2$, the solution for joint optimization boils to that obtained under time-sharing with $\zeta=1$, which procures a full array gain of $\mathcal{O}\left((N_1 + N_2)^2\right)$. When both the IRSs align with both the UEs, i.e., under event $\mathcal{A}$, due to the non-availability of closed-form expressions for $\phi_1^{\mathrm{opt}}$ and $\phi_2^{\mathrm{opt}}$, we do not have an explicit SE-scaling law. However, it is expected that the scaling of the array gain in this case lies between $\mathcal{O}\left(N_1^2 + N_2^2\right)$ and $\mathcal{O}\left((N_1+N_2)^2\right)$. Finally, when IRS-Y does not align with UE-$k$, the array gain again scales only as $\mathcal{O}(N_1^2)$. We can have similar interpretations at MO-Y as well. Based on the above, we make the following observations:
\begin{itemize}[leftmargin=*]
\item The sum-SE of MO-X scales at least as $\mathcal{O}(\log_2(N_1^2))$ in all cases, which is due to the array gain that is obtained in the absence of OOB IRSs. Thus, in general, IRS deployed by one MO \emph{does not degrade} the achievable SE at other MOs.
\item When both IRSs align to their respective OOB UEs, the best possible SE of an MO can potentially scale as $\mathcal{O}\left(\log_2(N_1+N_2)^2\right)$.  However, this is not simultaneously achievable at both MOs, as noted in the discussion following \eqref{eq:inconsistent_eq_case_A_2} in Sec.~\ref{subsec: Event A}. Further, even if it were possible to satisfy \eqref{eq:inconsistent_eq_case_A_1} and \eqref{eq:inconsistent_eq_case_A_2} simultaneously and achieve a sum-SE of $\mathcal{O}\left(\log_2(N_1+N_2)^2\right)$ at both MOs, we will next show that the resultant gain in the overall SE is small because of the low-probability nature of both IRSs aligning to OOB~UEs. 
\end{itemize}

\section{Quantifying the Effect of Out-of-Band IRSs}
\label{sec:Loss_quantify}
    In the previous section, we characterized the ergodic sum-SE of a system with $2$ MOs, each optimizing an IRS to serve its UEs. We analyzed three schemes that allow different degrees of cooperation between MOs. However, from a practical viewpoint, it is helpful to \textcolor{black}{explicitly quantify the gain/loss in the ergodic SE with/without OOB IRSs and with/without cooperation between MOs.} To that end, considering one of the MOs, say MO-X, we present the following result.
\begin{figure*}
\textcolor{black}{
\begin{equation}
\!\!\triangle \langle R_{\text{X}} \rangle_{\textrm{OOB}} = \frac{1}{K}\!\sum_{k=1}^K\!\frac{L_{k2}}{N_2}\!\log_2\!\left(\!1\!+\!\left\{\!\left(\!\frac{N_2^2}{L_{k2}}\beta^{(2)}_{X,k} \!
    	 +\! \mathbbm{1}_{\{\mathsf{CO} = 1\}} \!\frac{\pi N_1N_2(f(L_{k1}))^2\sqrt{\beta^{(1)}_{X,k}\beta^{(2)}_{X,k}}} {2\sqrt{L_{k1}L_{k2}}}\right)\!\!\Bigg/\!\!\left(\frac{\sigma^2}{P}\!+\!\frac{N_1^2}{L_{k1}}(g(L_{k1}))^2\beta^{(1)}_{X,k}\right)\!\right\}\!\right). \label{eq1_MO_X_SE_w_wo_OOB_IRS}
\end{equation}
\hrule}
\end{figure*}
\begin{figure*}
\textcolor{black}{
\begin{equation}
\triangle \langle R_{\text{X}} \rangle^{\textrm{L-SNR}}_{\textrm{OOB}} \approx \frac{1}{K}\sum_{k=1}^K\frac{1}{\ln(2)}\left({N_2}\beta^{(2)}_{X,k} 
    	 +\mathbbm{1}_{\{\mathsf{CO} = 1\}} \frac{\pi N_1(f(L_{k1}))^2\sqrt{\beta^{(1)}_{X,k}\beta^{(2)}_{X,k}}} {2}\cdot\sqrt{\frac{L_{k2}}{L_{k1}}}\right)\frac{P}{\sigma^2}. \label{eq3_MO_X_SE_w_wo_OOB_IRS}
\end{equation}
\hrule}
\end{figure*}
\begin{figure*}
\textcolor{black}{
\begin{equation}
\triangle \langle R_{\text{X}} \rangle^{\textrm{H-SNR}}_{\textrm{OOB}} \approx \frac{1}{K}\sum_{k=1}^K\frac{L_{k2}}{N_2}\log_2\left(1+ \left(\frac{N_2}{N_1}\right)^2\frac{L_{k1}}{L_{k2}}\cdot\frac{\beta^{(2)}_{X,k}}{(g(L_{k1}))^2\beta^{(1)}_{X,k}} + \mathbbm{1}_{\{\mathsf{CO} = 1\}}\cdot\frac{\pi N_2}{2 N_1}\cdot\frac{(f(L_{k1}))^2\sqrt{\beta^{(2)}_{X,k}L_{k1}}}{(g(L_{k1}))^2\sqrt{\beta^{(1)}_{X,k}L_{k2}}}\right). \label{eq2_MO_X_SE_w_wo_OOB_IRS}
\end{equation}
\hrule}
\end{figure*}
\color{black}
    \begin{theorem}\label{theorem-2}
       Under the SV channel model in the mmWave bands, under RR scheduling, the maximum gain in the ergodic sum-SE of MO-X  
       \begin{enumerate}[leftmargin=*]
       \item with, versus without, OOB IRS-Y  is given by~\eqref{eq1_MO_X_SE_w_wo_OOB_IRS}.
       \item with cooperation (i.e., jointly optimize/time-share the IRSs) versus no cooperation between MOs in the presence of the OOB IRS-Y is bounded as
               \begin{equation}\label{rate_gain_optimiz}
           \triangle \langle R_X \rangle_{\textrm{CO}} \leq \frac{1}{K} \sum\nolimits_{k=1}^K\frac{ L_{k2}}{N_2} \log_2\Bigl(1+\Psi(L_{k1})\frac{\pi}{4}\Bigl), 
        \end{equation} where $\Psi(L_{k1}) \triangleq (f(L_{k1}))^2\Big/g(L_{k1})$.
        \end{enumerate}
    \end{theorem}
    \color{black}
    \begin{proof} We prove the two statements seperately below.
    \subsubsection{Gain with versus without OOB IRS} To bound the gain in sum-SE obtained by MO-X with and without the OOB IRS, we consider two cases: $1)$ the OOB IRS is present, and it coherently adds the signals at the UEs of MO-X in all time slots if cooperation is allowed, and $2)$ the OOB IRS is absent. Let the sum-SE $\langle R_X^{\zeta} \rangle_{\textrm{TS}}$ given in~\eqref{final_rate_operator_X} with and without the OOB IRS be denoted by $\langle R_X^{\zeta} \rangle_{\textrm{W-IRS}}$ and $\langle R_X^{\zeta} \rangle_{\textrm{WO-IRS}}$, respectively. Then the maximum gain in SE is given by
    \begin{equation}
    \triangle \langle R_{\textrm{X}}\rangle_{\textrm{OOB}} \triangleq \langle R_X^{\zeta} \rangle_{\textrm{W-IRS}}\Big\vert_{\zeta=1} - \langle R_X^{\zeta} \rangle_{\textrm{WO-IRS}},
    \end{equation}
    where $\zeta=1$ captures that the overall phase shifts at both the IRSs are used to coherently add the signals at the UEs served by MO-X in all time slots. Then, substituting for the resulting values in the above equation and noting that the sum-SE without the OOB IRS follows by substituting $N_2 = 0$ in~\eqref{final_rate_operator_X} and recognizing that the pre-log term $L_{k2}/N_2$ is unity in the absence of the OOB IRS, we obtain~\eqref{eq1_MO_X_SE_w_wo_OOB_IRS}.    \subsubsection{Gain with versus without cooperation} We first consider time-sharing. Let the sum-SE at MO-X with and without cooperation be denoted as $\langle R_{\textrm{X}} \rangle^{\textrm{TS}}_{\textrm{W-CO}}$ and $\langle R_{\textrm{X}} \rangle^{\textrm{TS}}_{\textrm{WO-CO}}$, respectively, i.e., from~\eqref{final_rate_operator_X}, $\langle R_{\textrm{X}} \rangle^{\textrm{TS}}_{\textrm{W-CO}} = \langle R_{\textrm{X}}^{\zeta} \rangle_{\textrm{TS}}$ with $\mathsf{CO}=1$, and $\langle R_{\textrm{X}} \rangle^{\textrm{TS}}_{\textrm{WO-CO}} =  \langle R_{\textrm{X}}^{\zeta} \rangle_{\textrm{TS}}$ with $\mathsf{CO}=0$. Then, the maximum gain in the sum-SE of MO-X with versus without cooperation is
    \begin{align}
    \triangle \langle R_{\textrm{X}}\rangle_{\textrm{TS}} &\triangleq \langle R_{\textrm{X}} \rangle^{\textrm{TS}}_{\textrm{W-CO}}\Big\vert_{\zeta=1} -  \langle R_{\textrm{X}} \rangle^{\textrm{TS}}_{\textrm{WO-CO}}\\
    &\hspace{-1.6cm}=\frac{1}{K}\sum_{k=1}^K\! \frac{L_{k2}}{N_2}\log_2\left(\!\!1\!+\!\!\frac{\pi}{2}\dfrac{\frac{N_1N_2} {\sqrt{L_{k1}L_{k2}}} (f(L_{k1}))^2\sqrt{\beta^{(1)}_{X,k}\beta^{(2)}_{X,k}}}{\frac{\sigma^2}{P}\!+\!\left(\frac{N_1^2}{L_{k1}}(g(L_{k1}))^2\beta^{(1)}_{X,k}\!+\!\frac{N_2^2}{L_{k2}}\beta^{(2)}_{X,k}\right)}\!\right) \nonumber \\
       	& \stackrel{(b)}{\lessapprox} \frac{1}{K}\sum\nolimits_{k=1}^K \frac{L_{k2}}{N_2}\log_2\left(1+\Psi(L_{k1})\frac{\pi}{4}\right) \label{eq:loss_bound_CO},
         \end{align}
         where $\Psi(L_{k1})$ is as defined in the theorem and in $(b)$ we used a high SNR approximation and the fact that $\Bigl(({N_1}/{\sqrt{L_{k1}}})g(L_{k1})\sqrt{\beta^{(1)}_{X,k}} -({N_2}/{\sqrt{L_{k2}}})\sqrt{\beta^{(2)}_{X,k}}{\Bigl)}^2 \geq 0$. 
         
         Next, under the joint optimization scheme, let the SE with and without cooperation be  $\langle R_{\textrm{X}} \rangle^{\textrm{JO}}_{\textrm{W-CO}}$ and $\langle R_{\textrm{X}} \rangle^{\textrm{JO}}_{\textrm{WO-CO}}$, respectively, i.e., from~\eqref{eq_Joint_optim_full_ergodic_SE_MO_X}, $\langle R_{\textrm{X}} \rangle^{\textrm{JO}}_{\textrm{W-CO}} = \langle R_{\textrm{X}} \rangle_{\textrm{JO}}$, and $\langle R_{\textrm{X}} \rangle^{\textrm{JO}}_{\textrm{WO-CO}} =  \langle R_{\textrm{X}}^{\zeta} \rangle_{\textrm{TS}}$ with $\mathsf{CO}=0$. Then, the maximum gain in the SE is 
         \begin{align}
         \triangle \langle R_{\textrm{X}} \rangle_{\textrm{JO}} &= \langle R_{\textrm{X}} \rangle^{\textrm{JO}}_{\textrm{W-CO}} - \langle R_{\textrm{X}} \rangle^{\textrm{JO}}_{\textrm{WO-CO}}\\
         &\stackrel{(c)}{\leq} \langle R_{\textrm{X}} \rangle^{\textrm{TS}}_{\textrm{W-CO}}\Big\vert_{\zeta=1} -  \langle R_{\textrm{X}} \rangle^{\textrm{TS}}_{\textrm{WO-CO}}, \label{eq_JO_diff_final_temp}
         \end{align}
         where in $(c)$, we use Theorem~\ref{thm:theorem-1} that the SE achieved by jointly optimizing the IRSs is upper bounded by the SE when the IRSs are optimized for the UE served by MO-X in all time slots. Finally, we note that~\eqref{eq_JO_diff_final_temp} can be characterized as given in~\eqref{eq:loss_bound_CO}. Thus, under both schemes, the gain with versus without cooperation can be unified into a single expression, $\triangle \langle R_X \rangle_{\textrm{CO}}$, in~\eqref{rate_gain_optimiz}. This completes the proof.
        \end{proof}
        \color{black}
        From~\eqref{eq1_MO_X_SE_w_wo_OOB_IRS} of Theorem~\ref{theorem-2}, we observe that the gain in the sum-SE is strictly non-negative. Thus, the ergodic sum-SE at MO-X can only improve in the presence of an OOB IRS.
In particular, we also make the following observations:
    \begin{itemize} [leftmargin=*]
    \item \textbf{Gain with OOB IRS at Low-SNR}: In the low SNR regime, since $\dfrac{\sigma^2}{P}\gg 1$, we have $\dfrac{\sigma^2}{P}+\dfrac{N_1^2}{L_{k1}}(g(L_{k1}))^2\beta^{(1)}_{X,k} \approx \dfrac{\sigma^2}{P}$. Using this in~\eqref{eq1_MO_X_SE_w_wo_OOB_IRS}, and also using $\log_2(1+x)\approx \dfrac{x}{\ln(2)}$ when $x \ll 1$, we obtain the simplified expression for the gain obtained by an in-band MO due to an OOB-IRS in~\eqref{eq3_MO_X_SE_w_wo_OOB_IRS}, at low-SNR scenarios. In this case, the gain increases linearly with the number of IRS elements, particularly those of the OOB IRS, and the SNR of operation. This behavior can be attributed to the fact that the OOB IRS enables the reception of additional copies of the signal at the UE (either coherently or incoherently, depending on the level of MO cooperation) whenever the OOB IRS aligns with the UE served by MO-X (an event occurring with probability $L_{k2}/N_2$). These additional signal paths enhance the \emph{average} received SNR at the in-band UE.
    \item \textbf{Gain with OOB IRS at High-SNR}: Here, since $\dfrac{\sigma^2}{P} \ll 1$, we have $\dfrac{\sigma^2}{P}+\dfrac{N_1^2}{L_{k1}}(g(L_{k1}))^2\beta^{(1)}_{X,k} \approx \dfrac{N_1^2}{L_{k1}}(g(L_{k1}))^2\beta^{(1)}_{X,k}$. Using this approximation, along with the fact that $\log_2(1+x) \approx \log_2(x)$ when $x\gg 1$, we get the simplified expression for the SE gain in the high-SNR regime as given in~\eqref{eq2_MO_X_SE_w_wo_OOB_IRS}. Contrary to the low-SNR regime, the gain at high SNR exhibits a unimodal behavior with respect to the number of elements at the OOB IRS. This arises because, although a larger OOB IRS can potentially deliver more signal copies when aligned with the in-band UE, the probability of such alignment decreases as the number of elements at the OOB IRS increases. Initially, the alignment probability is high, allowing the overall SE gain to improve with $N_2$. However, the alignment probability decreases as $N_2$ increases, and this outweighs the logarithmic increase in the gain. This leads to an overall reduction in the gain, resulting in the observed unimodal trend. Moreover, the gain does not scale significantly with SNR, since the in-band UE already enjoys high SE at high SNR. Thus, unless the OOB IRS offers substantial additional contribution (which is less likely as $N_2$ increases), only marginal improvements are observed.
            \item The best SE gain in~\eqref{rate_gain_optimiz} obtained by cooperation is directly proportional to $L_{k2}/N_2$, but depends weakly on $L_{k1}$ through the $\Psi(L_{k1})$ term, as shown in Table~\ref{table:L_vs_Psi}. However, the gain decreases as the number of OOB IRS elements $N_2$ increases.
    \end{itemize}
     \color{black}
           
    In the next section, we extend our results to a general setting where more than two operators co-exist and each deploys its own IRS to optimally serve its UEs. 
   
\begin{table}[t]
	\centering
		\captionof{table}{Variation of $\Psi(L)$ as a function of $L$.}
	\begin{tabular}{ |c|c|c|c|c|c|c|c|c|c|c| } 
		\hline
		$L$ & 1 & 2 & 5 & 10 & 25 & 40 \\ 
		\hline
		$\Psi(L)$ & 0.79 & 0.87 & 0.93 & 0.96 &0.97 & 0.98\\ 
		\hline
	\end{tabular}
	\label{table:L_vs_Psi}
	\vspace{-0.4cm}
\end{table}
     
          \section{Performance Analysis with $M>2$ MOs} \label{sec:M-BS & M-IRS System}
 
 We consider that $M$ MOs serve a given geographical area, and their respective BSs: \{BS-$1$, BS-$2$,\ldots, BS-$M$\} provide services to \{{$K_1$}, {$K_2$}, \ldots, {$K_M$}\} UEs at the same time over non-overlapping bands. Further, UE-$k_{l}$ denotes the $k$th UE served by the $l$th MO.
For simplicity, we let the number of elements in each IRS equal $N$. At any UE served by an in-band MO, due to the presence of $M-1$ OOB IRSs, $M$ different events arise, similar to Sec.~\ref{sec:optimization_mmWave_2-BS}, and we denote them by $\mathcal E_0,\ldots, \mathcal E_{M-1}$. Specifically, $\mathcal E_m$ is the event that exactly $m$ OOB IRSs align with UE-{$k_1$} on one of the OOB paths through them. Further, the event that an OOB IRS phase configuration aligns with an OOB path is independent across the OOB IRSs, so, the number of OOB IRSs, $m$, that contribute to the channel at any UE follows a binomial distribution, i.e., $m\sim\text{Bin}\left(M-1,\frac{L_{k1}}{N}\right)$. As a consequence, $\Pr(\mathcal E_m) = {M-1 \choose m}{\Bigl(\frac{ L_{k_1}}{N}\Bigl)}^m\Bigl(1-\frac{ L_{k_1}}{N}\Bigl)^{(M-1-m)}$, where $L_{k_1}$ is the number of resolvable OOB paths at UE-$k_1$ via an IRS, and here $L_{k1}\leq N$.
 Then, similar to~\eqref{final_eq_kth_user_Event_A}, conditioned on $\mathcal E_m$, the channel from BS-1 to UE-{$k_1$} is given by 
 \begin{equation*}
 	\vspace{-0.1cm}
 	h_{k_1}\!\! =\!\! \frac{N}{\sqrt{L_{k_1}}}\biggl\{\!\Bigl|\gamma^{(1)}_{1,\textrm{B}_1}\gamma^{(1)}_{1,{k_1}}\Bigl| e^{j\phi_1}\!+\!\sum_{i=2}^{m+1} \Bigl|\gamma^{(m_i)}_{l^*,\textrm{B}_1}\gamma^{(m_i)}_{l^*,{k_1}}\Bigl| e^{j (\phi_{m_i} \!+\! \phi_{m_{i1}})}\!\!\biggl\},
 \end{equation*} 
 where symbols have similar meanings as in Sec.~\ref{sec:optimization_mmWave_2-BS}. For example, $\gamma^{(1)}_{1,\textrm{B}_1}$ denotes the channel coefficient between the first IRS (superscript) and BS-1 (second subscript) along the first (dominant) path (first subscript); $\gamma^{(1)}_{1,{k_1}}$ denotes the channel coefficient between the first IRS and the $k_1$th UE along the first path. Also, $m_i \in \{2,\ldots,M\}$, such that $m_i \neq m_j$ when $i \neq j$, i.e., $\{m_i\}_{i=2}^{m+1}$ denotes the indices of the $m$ OOB IRSs for which some $l^*$th path aligns with UE-$k_1$.
  Further, $\gamma^{(m_i)}_{l^*,\textrm{B}_1}$ and $\gamma^{(m_i)}_{l^*,k_1}$ denote the coefficients of the channels from BS-$1$ to IRS-$m_i$ and IRS-$m_i$ to UE-$k_1$ that correspond to the $l^*$th OOB cascaded path via the aligning IRS-$m_i$ which contributes to the recieved signal at UE-$k_1$. 
  Then, we model $\gamma^{(m_i)}_{l^*,\textrm{B}_1} \!\sim\! \mathcal{CN}(0,\beta_{\textrm{B}_1})$, $\gamma^{(m_i)}_{l^*,{k_1}} \! \sim\! \mathcal{CN}(0,\beta_{{k_1}})$, where $\beta_{\textrm{B}_1}$ and $\beta{_{k_1}}$ are the path losses in the BS-$1$ to IRS-$1$ and IRS-$1$ to UE-{$k_1$} links, respectively.\footnote{For simplicity of exposition, the path losses are equal across IRSs~\cite{Lin_Bai_Access_2017}.}  Finally, $\phi_{m_{i1}}$ is the phase difference of the matching cascaded OOB path at UE-$k_1$ via the $m_i$th IRS and the in-band path, and $\phi_{m_i}$ is an overall extra phase applied at IRS-$m_i$. Next, we analyze the ergodic sum-SE of MO-$1$ for different schemes as discussed in the $2$-MO case, which entails varying degrees of cooperation among the MOs.
 \vspace{ -0.2 cm}
 \subsection{Time-sharing of the IRSs with MO Cooperation} 
 In this scheme, under event $\mathcal{E}_m$, for a $\zeta_1$ fraction of time slots, all the $m$ matching OOB IRSs coherently add the signals at UE-$k_1$ and in the other $1-\zeta_1$ fraction of slots, UE-$k_1$ receives an incoherent addition of signals from the OOB IRSs. Then,
 similar to Sec.~\ref{subsec: Event A}, to maximize $|h_{k_1}|^2$ at UE-$k_1$ in the $\zeta_1$ fraction of slots, we set $\phi_{m_i} = -\phi_{m_{i1}}$, and $\phi_1=0$. The overall channel coefficient at UE-$k_1$ under event $\mathcal{E}_m$ is
 \begin{equation}\label{eq_h_k_1}
 	\vspace{-0.1cm}
 	h^{(\zeta_1)}_{k_1} \stackrel{d}{=} \frac{N}{\sqrt{L_{k1}}}\biggl\{\Bigl|\gamma^{(1)}_{1,\textrm{B}_1}\gamma^{(1)}_{1,{k_1}}\Bigl|+\sum_{i=2}^{m+1} \Bigl|\gamma^{(m_i)}_{l^*,\textrm{B}_1}\gamma^{(m_i)}_{l^*,{k_1}}\Bigl|\biggl\}, 
 \end{equation}
 and its average gain in~\eqref{eq_simplify_corr_proof} is obtained as $\mathbb{E}\Bigl[\left|h^{(\zeta_1)}_{k_1}\right|^2\Bigl] =$
  \begin{equation*}
 \frac{N^2}{L_{k_1}}\beta_{\textrm{B}_1,k_1} \biggl\{(g(L_{k1}))^2\!+\!m\Bigl(1 
 \!+\!\frac{\pi}{2}(f(L_{k1}))^2\!+\!\frac{\pi^2}{8}\frac{(m-1)}{2}\Bigl)\biggl\},
 \vspace{-0.1cm}
 \end{equation*}
where $\beta_{\textrm{B}_1,k_1} \triangleq \beta_{\textrm{B}_1}\beta_{{k_1}}$. For other $(1-\zeta_1)$ fraction of slots,
 \begin{equation}
\!\!\!\!\!h_{k_1}^{(1-\zeta_1)} \!\triangleq\! \frac{N}{\sqrt{L_{k_1}}}\biggl\{\!\Bigl|\gamma^{(1)}_{1,\textrm{B}_1}\gamma^{(1)}_{1,{k_1}}\Bigl|+\!\sum_{i=2}^{m+1} \Bigl|\gamma^{(m_i)}_{l^*,\textrm{B}_1}\gamma^{(m_i)}_{l^*,{k_1}}\Bigl| e^{j \phi_{m_{i1}}}\!\!\biggl\}, \!\!\!\!
 \end{equation}
 and its average gain can be similarly obtained as
 \begin{equation}
 \mathbb{E}\Bigl[\left|h^{(1-\zeta_1)}_{k_1}\right|^2\Bigl] = ({N^2}/{L_{k_1}})\beta_{\textrm{B}_1,k_1} \biggl\{(g(L_{k1}))^2\!+\!m\biggl\}.
 \end{equation}
 Then, by using Jensen's approximation, the ergodic sum-SE of MO-X conditioned on event $\mathcal{E}_m$ is given in~\eqref{eq_event_Em_TS_M_gg_2}. 
 
 \subsection{Joint-Optimization of IRSs with MO Cooperation}
In the joint optimization scheme, the overall phase angles at the IRSs, i.e., $\phi_{m_i}$, are chosen such that the weighted sum-SE of all UEs scheduled by every MO in a time slot is maximized. In this case, the SE of each MO under different events can be characterized similarly to the previous section. Consequently, the overall sum-SE of any given MO can be obtained similar to~\eqref{eq_Joint_optim_full_ergodic_SE_MO_X}, and we omit the details due to space constraints.
  \begin{figure*}
 	\vspace{-0.8cm}
 	\begin{equation}\label{eq_simplify_corr_proof}
 		\mathbb{E}\Bigl[\left|h_{k_1}^{(\zeta)}\right|^2\Bigl] = \frac{N^2}{L_{k_1}} \Biggl\{\mathbb{E}\biggl[{\Bigl|\gamma^{(1)}_{1,\textrm{B}_1}\gamma^{(1)}_{1,k_1}\Bigl|}^2+\sum_{i=2}^{m+1}{\Bigl|\gamma^{(m_i)}_{l^*,\textrm{B}_1}\gamma^{(m_i)}_{l^*,k_1}\Bigl|}^2+2\sum_{i=2}^{m+1}\Bigl|\gamma^{(m_i)}_{l^*,\textrm{B}_1}\gamma^{(m_i)}_{l^*,k_1}\gamma^{(1)}_{1,\textrm{B}_1}\gamma^{(1)}_{1,k_1}\Bigl|  +\sum_{i,j=2\atop {i \neq j}}^{m+1}\Bigl|\gamma^{(m_i)}_{l^*,\textrm{B}_1}\gamma^{(m_i)}_{l^*,k_1}\gamma^{(m_j)}_{l^*,\textrm{B}_1}\gamma^{(m_j)}_{l^*,k_1}\Bigl|\biggl]\Biggl\}.
 		\vspace{0.1cm}
 	\end{equation}
 	\hrule
 \end{figure*}
  \begin{figure*}
  \vspace{-0.5cm}
 \begin{multline}\label{eq_event_Em_TS_M_gg_2}
 	\langle R_{k_1}^{(\zeta_1)}|\mathcal E_m \rangle_{\textrm{TS}} \approx  \zeta_1 \log_2\biggl(1+\frac{P}{\sigma^2}\frac{N^2}{L_{k_1}}\beta_{\textrm{B}_1,k_1} \biggl\{(g(L_{k1}))^2  + m\Bigl(1 
 	+\frac{\pi}{2}(f(L_{k1}))^2+\frac{\pi^2}{8}\frac{(m-1)}{2}\Bigl)\biggl\}\biggl) \\ + (1-\zeta_1) \cdot \log_2\biggl(1+\frac{P}{\sigma^2}\frac{N^2}{L_{k_1}}\beta_{\textrm{B}_1,k_1} \biggl\{(g(L_{k1}))^2  + m\biggl\}\biggl). 
 \end{multline}
 \hrule
 \end{figure*}
 
 \subsection{No cooperation among the MOs}
When the MOs do not cooperate, the overall IRS phase-shifts are set as $\{\phi_{1},\phi_{m_i}\}_{i=2}^{m+1} = 0$. So, the channel becomes 
 \begin{equation*}
 	\vspace{-0.2cm}
 	h_{k_1} = \frac{N}{\sqrt{L_{k_1}}}\biggl\{\!\Bigl|\gamma^{(1)}_{1,\textrm{B}_1}\gamma^{(1)}_{1,{k_1}}\Bigl|+\sum_{i=2}^{m+1} \Bigl|\gamma^{(m_i)}_{l^*,\textrm{B}_1}\gamma^{(m_i)}_{l^*,{k_1}}\Bigl| e^{j \phi_{m_{i1}}}\biggl\},
	\vspace{-0.1cm}
 \end{equation*} 
for which, we have $\mathbb{E}\Bigl[|h_{k_1}|^2\Bigl] = \frac{N^2}{L_{k_1}}\beta_{\textrm{B}_1,k_1} \biggl\{(g(L_{k1}))^2\!+\!m\biggl\}$. Then the ergodic SE of MO-$1$, under event $\mathcal{E}_m$, is
\vspace{-0.1cm}
\begin{equation*}
		\langle R_{k_1}|\mathcal E_m\rangle_{\mathrm{NCO}} \approx \log_2\biggl(1+\frac{P}{\sigma^2}\frac{N^2}{L_{k1}}\beta_{\textrm{B}_1,k_1}\Bigl\{(g(L_{k1}))^2+m\Bigl\}\biggl).
\vspace{-0.1cm}
\end{equation*}

 We next characterize the overall ergodic sum-SE of a MO (say MO-$1$) when $M>2$ MOs coexist, similar to Theorem~\ref{thm:theorem-1}.

 \begin{theorem}\label{corollary:1}
    Under the SV channel model in the mmWave bands, when $M>2$ MOs control an IRS each to serve its subscribed UEs, the ergodic sum-SE of MO-$1$ under RR scheduling is characterized as:  	\begin{enumerate}[leftmargin=*]
 		\item Optimization of IRSs with time sharing where $\zeta_1 T_c$ time slots alloted to MO-X and $(1-\zeta_1)T_c$ time slots are alloted to other MOs: $\langle R_1^{\zeta_1} \rangle_{\textrm{TS}}$, as given in~\eqref{thm_eq_corr_1} with $\mathsf{CO}=1$.
		\item Joint-optimization of IRSs with MO cooperation: The ergodic rate for MO-$1$ can be obtained similar to~\eqref{eq_Joint_optim_full_ergodic_SE_MO_X} with $\phi_{m_i}^{\mathrm{opt}}$ is determined similar to Algorithm~\ref{alg:One_iteration_method}. Further, an upper bound on the SEs $\langle R_1 \rangle_{\textrm{JO}}$ without relying on $\phi_{m_i}^{\mathrm{opt}}$ can be obtained as $\langle R_1 \rangle_{\textrm{JO}} \leq \langle R_1^{\zeta_1} \rangle_{\textrm{TS}}\Big\vert_{\zeta_1=1}$, 
	where $\langle R_X^{\zeta_1} \rangle_{\textrm{TS}}$ is given in~\eqref{thm_eq_corr_1} with $\mathsf{CO}=1$.
 		\item No MO Cooperation: $\langle R_1 \rangle_{\textrm{NCO}}$ as in~\eqref{thm_eq_corr_1} with $\mathsf{CO}=0$.
 	\end{enumerate}
 	\begin{figure*}
 		\begin{multline} \label{thm_eq_corr_1}
 			 \langle R_{1}^{(\zeta_1)} \rangle_{\textrm{TS}} \approx  \frac{1}{K_1}\sum_{k_1=1}^{K_1}\sum_{m=0}^{M-1} {M-1 \choose m}{\biggl(\frac{ L_{k_1}}{N}\biggl)}^m{\biggl(1-\frac{ L_{k_1}}{N}\biggl)}^{\!(M\!-\!m\!-\!1)}\hspace{-0.4cm}\times\left\{\zeta_1\log_2\biggl(1+\frac{P}{\sigma^2}\frac{N^2}{L_{k_1}}\beta_{\textrm{B}_1,k_1} \biggl\{(g(L_{k1}))^2+m + \right. \\
 			\left. \mathbbm{1}_{\{\mathsf{CO} = 1\}} \ m\frac{\pi}{2}\Bigl((f(L_{k1}))^2 +\frac{\pi(m-1)}{8}\Bigl)\biggl\}\biggl)+ (1-\zeta_1)\cdot\log_2\biggl(1+\frac{P}{\sigma^2}\frac{N^2}{L_{k_1}}\beta_{\textrm{B}_1,k_1} \biggl\{(g(L_{k1}))^2  + m\biggl\}\biggl)\right\},
 		\end{multline}
 		\hrule
 	\end{figure*}
	\begin{figure*}
 		\begin{equation} \label{eq1_MO_X_SE_w_wo_OOB_IRSs_M_ge_2}
 			\!\!\Delta \langle R_{1} \rangle_{\textrm{OOB}} \approx  \frac{1}{K_1}\sum_{k_1=1}^{K_1}\sum_{m=1}^{M-1} {M-1 \choose m}{\biggl(\frac{ L_{k_1}}{N}\biggl)}^m{\biggl(1-\frac{ L_{k_1}}{N}\biggl)}^{\!(M\!-\!m\!-\!1)}\hspace{-0.7cm}\times\log_2\!\left(\!1+\frac{m+\!\mathbbm{1}_{\{\mathsf{CO} = 1\}}m\frac{\pi}{2}\Bigl((f(L_{k1}))^2 +\frac{\pi(m-1)}{8}\Bigl)}{(g(L_{k1}))^2}\right),
 		\end{equation}
		\vspace{-0.1cm}
 		\hrule
 	\end{figure*}
	 \begin{figure*}
	 \vspace{0.1cm}
 \begin{equation}\label{eq:M_IRS_M_BS_eqn_2}
 		\triangle \langle R_1 \rangle_{\textrm{CO}} \!\leq \frac{1}{{K_1}}\!\sum_{k_1=1}^{K_1}\!\sum_{m=1}^{M-1}\! {M-1 \choose m}{\biggl(\frac{ L_{k_1}}{N}\biggl)}^m\!\!{\biggl(1-\frac{ L_{k_1}}{N}\biggl)}^{\!\!(M\!-\!m\!-\!1)}  \!\!\times
 		\log_2\biggl(1+\frac{\pi\sqrt{m}}{4}\left\{\Psi(L_{k1}) +\frac{\pi(m-1)}{8g(L_{k1})}\right\}\biggl),
 	\end{equation}
	\hrule
	\vspace{-0.5cm}
 \end{figure*}
 \end{theorem}
 \vspace{-0.4 cm}
 \begin{proof}
 By the law of total expectation, $\langle R_{k_1}\rangle = \sum_m\langle R_{k_1}|\mathcal E_m \rangle \Pr(\mathcal E_m)$, and under RR scheduling, we note $\langle R_{1} \rangle \triangleq \frac{1}{K_1}\sum_{k_1=1}^{K_1} \langle R_{k_1} \rangle$. Using the values of $\langle R_{k_1} \rangle$ under the three schemes in $\langle R_{1} \rangle$ completes the proof.
 \end{proof}
 \vspace{-0.2cm}
 Next, similar to Theorem 2, we can characterize the gain in the sum-SE due to the presence of the OOB IRSs over that in the absence of OOB IRSs as well as the gain due to cooperation over the no-cooperation case. We illustrate this in the following result, which shows that cooperation offers only a marginal improvement in the sum-SE.
 \begin{theorem}\label{corollary:2}
 Under the SV channels in mmWaves and RR scheduling, the maximum gain in the ergodic sum-SE of MO-1  
       \begin{enumerate}[leftmargin=*]
       \item with vs. without OOB IRSs is approximately given by~\eqref{eq1_MO_X_SE_w_wo_OOB_IRSs_M_ge_2}.
       \item with cooperation (i.e., jointly optimize/time-share the IRSs) vs. no cooperation between MOs in the presence of the OOB IRSs is bounded as in~\eqref{eq:M_IRS_M_BS_eqn_2}.        \end{enumerate}
 \end{theorem}
 \begin{proof}
It is similar to Theorem~\ref{theorem-2}. We omit for brevity.
	 \end{proof}
Considering the $m=1$ term in \eqref{eq1_MO_X_SE_w_wo_OOB_IRSs_M_ge_2} and \eqref{eq:M_IRS_M_BS_eqn_2}, and comparing them with \eqref{eq2_MO_X_SE_w_wo_OOB_IRS} and \eqref{rate_gain_optimiz} respectively, we see that the gain in SE due to the presence of $M$ IRSs scales approximately as $M-1$ times the gain in the 2-MO case. This is because there are $M-1$ OOB IRSs that can potentially align with a given UE. In addition, we can obtain further gains in the SE, captured by the summands corresponding to $m=2$ to $m=M-1$, when $m\ge 2$ OOB IRSs happen to be aligned to the UE. However, the event that $m$ IRSs align with a UE occurs with exponentially lower probability due to the $(L/N)^m$ term. Thus, the presence of $M$ MOs does not degrade the sum SE of a given operator; in fact, it provides a gain in the sum SE that increases at least linearly with $M$. 

 In the next section, we numerically illustrate our findings via Monte Carlo simulations.    
 
 \vspace{-0.25cm}   
		\section{Numerical Results and Discussion}\label{sec:numerical_sections}
		We first illustrate the results for $2$ MOs as in Sec.~\ref{sec:optimization_mmWave_2-BS} and~\ref{sec:Loss_quantify}.
		\vspace{-0.5cm}
    \subsection{$2$-MO \& $2$-IRS System}\label{sec:2-MO-numerical-results}
BS-X and BS-Y are located at $(0,200)$, and $(200,0)$ (in meters), and IRS-X and IRS-Y are located at $(0,0)$ and $(200,200)$, respectively. All UEs are uniformly located in a rectangular region with diagonally opposite corners $(0,0)$ and $(200,200)$. The path loss is modeled as $\beta = C_0\left(d_0/d\right)^\alpha$, where $C_0=-60$ dB is the path loss at the reference distance $d_0=1$ m, $d$ is the distance between nodes, and $\alpha$ is the path loss exponent~\cite{RuiZhang_TWC_2019}. 
	\textcolor{black}{We use $\alpha=2$ for both BS-IRS, and the IRS-UE paths~\cite{Wang_TVT_2020}.} We use RR scheduling to serve  $K\!=\!Q\!=\!10$ UEs over $1000$ time slots by the respective MOs. \\
\begin{figure}[t]
\vspace{-0.2cm}
\centering
\includegraphics[width=\linewidth]{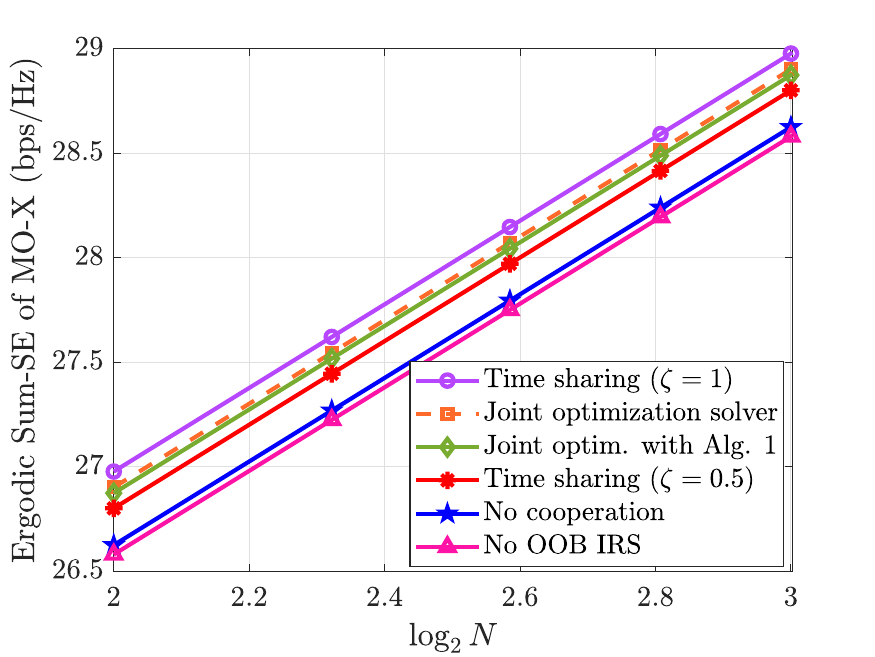}
\caption{Ergodic sum-SE of MO-X vs. $\log_2 N$ conditioned on Event $\mathcal{A}$ at $C_0\gamma=150$ dB and \textcolor{black}{$L_{k1}=1, L_{k2}=10$.}}
\label{fig:In_band_UE_vs_log_2N}
  \vspace{-0.2cm}
\end{figure}
\begin{figure}[t]
	\centering
	\includegraphics[width=\linewidth]{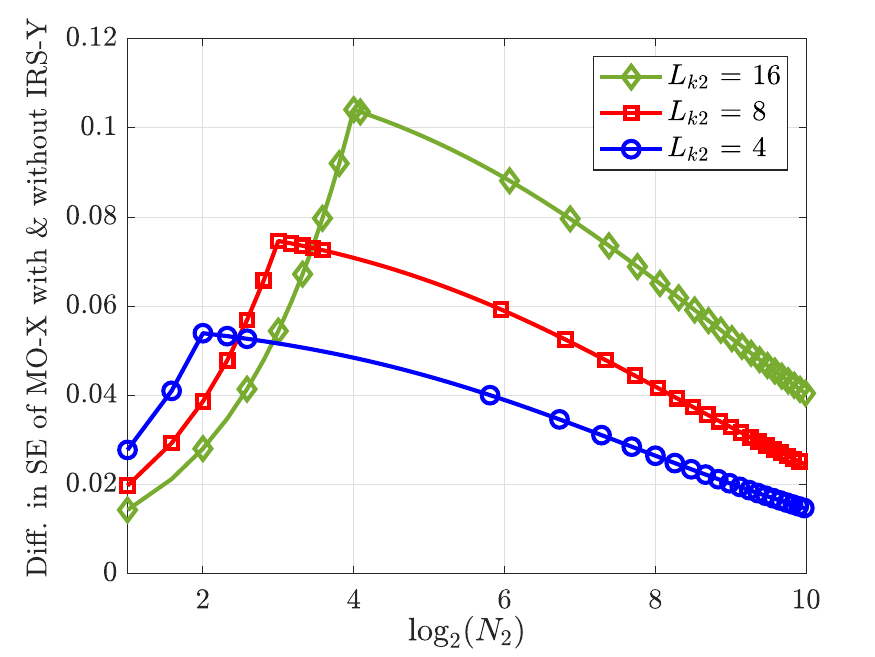}
	\caption{Diff. in sum-SE of MO-X with and without OOB IRS vs. $\log_2(N_2)$.}
	\label{fig:rate_difference_iib_not_changing}
	  \vspace{-0.2cm}
\end{figure}
\begin{figure}[!t]
\centering
\includegraphics[width=\linewidth]{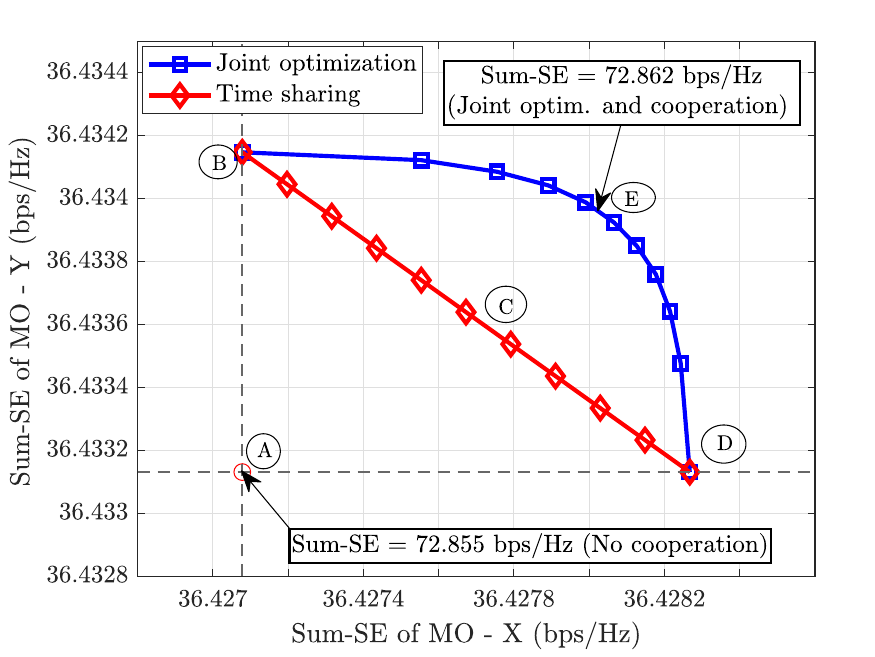}
\caption{Rate Region of the MOs at \textcolor{black}{$L_{k1} = 1, L_{k2} = 5$, $L_{q1} = 8, L_{q2} = 1$}.}
\label{fig:rate_region}
\vspace{-0.2cm}
\end{figure}
In Fig.~\ref{fig:In_band_UE_vs_log_2N},
we plot the achievable ergodic sum-SE of MO-X vs. $\log_2 N$ (where $\!N=\!N_1\!=\!N_2$), under event $\mathcal{A}$, and study the performance of the three schemes described in Sec.~\ref{subsec: Event A}. \textcolor{black}{The SE, when jointly optimal overall IRS phases are used to maximize the equal-weighted sum-SE of UEs scheduled by the MOs using a general high complexity off-the-shelf solver using the \texttt{findpeaks} function of MATLAB (curve labeled \texttt{Joint optimization solver}) nearly overlaps with that obtained using the low-complexity single-iteration Algorithm~\ref{alg:One_iteration_method} (curve labeled \texttt{Joint optim. with Alg. 1}). This shows that the proposed single-iteration Newton's algorithm is a practically viable solution, offering near-optimal performance with significantly reduced complexity. This effectiveness is largely due to the carefully chosen initialization strategy (given in lines $1$--$5$ of Algorithm~\ref{alg:One_iteration_method}), which provides provable convergence guarantees~\cite[Theorem $9.1$]{chong2013introduction}.} 
Further, the SE using the joint optimization scheme is only slightly inferior to the time-sharing scenario with $\zeta = 1$ (curve labeled \texttt{Time sharing $\zeta=1$}). 
This is because no UE gets the full array gain of $\mathcal{O}((N_1+N_2)^2)$ in any time slot under a joint-optimization scheme. 
The performance obtained by MO-X with time-sharing ($\zeta = 1$) is about $0.5$ bps/Hz higher than that obtained without an OOB IRS (curve labeled \texttt{No OOB IRS}). This is because, under event $\mathcal{A}$, the OOB IRS approximately doubles the SNR at the UE (thereby improving the SE by $1$ bps/Hz) when the UE is closer to the OOB IRS than the in-band IRS. On the other hand, when the UE is closer to the in-band IRS than the OOB IRS, the SNR is nearly the same as that in the absence of the OOB IRS. These two events are equally likely under the simulation setup considered, hence, the average gain in SE through cooperation is about $0.5$ bps/Hz. The performance obtained by MO-X with no cooperation nearly matches with that obtained in the absence of the OOB IRS (the bottom two curves), because the SNR gain from the OOB IRS under event $\mathcal{A}$ is negligible when the overall phase of the IRS is arbitrary. More importantly, the OOB IRS does not degrade the SE even if the MOs do not cooperate. Finally, the ergodic sum-SE of \textcolor{black}{MO-X} is log-quadratic in $N$ in all scenarios, thus, the array gain from IRS-X is always obtained.
 
\indent Next, in Fig.~\ref{fig:rate_difference_iib_not_changing}, for a fixed number of elements at IRS-X (at $N_1=64$), we plot the difference between the ergodic sum-SE of MO-X obtained in the presence and absence of the OOB IRS-Y vs. the number of OOB IRS elements (in the log-domain) as a function of the number of OOB paths, $L_{k2}$. To capture the maximum possible difference, we consider that whenever IRS-Y aligns to the in-band UE of MO-X, both MOs cooperate and implement the time-sharing scheme with $\zeta=1$. 
Then, for a given $L_{k2}$, we observe that the difference is non-negative, and further, this gain in the SE due to the presence of OOB IRS-Y is an unimodal function in $N_2$ with the peak occurring at $N_2 = L_{k2}$. This is in line with the theoretical expression given by~\eqref{eq2_MO_X_SE_w_wo_OOB_IRS} in 
Theorem~\ref{theorem-2}. 
Intuitively, at smaller values of $N_2$, with high probability, IRS-Y aligns with the in-band UE of MO-X; so, when $N_2$ increases, the overall SNR increases for MO-X. However, for larger values of $N_2$, the probability that the IRS-Y aligns to MO-X's UE becomes small, in turn causing the change in the SE with and without an OOB IRS to decrease. Finally, as $L_{k2}$ increases, the gain again increases because the probability of IRS-Y aligning to MO-X increases, which further enhances the channel gain at in-band UEs of MO-X.
 Thus, an OOB IRS benefits MO-X more when there are many paths via the OOB IRS at the UEs served by MO-X.\\ 
\indent In Fig.~\ref{fig:rate_region}, we plot the achievable rate regions of the two MOs (\textcolor{black}{normalized by the bandwidths}) for $N_1 = N_2 = 256$ and $C_0\gamma=150$ dB under two different schemes: $1)$ time-sharing (corresponding to the curve \encircle{B}-\encircle{C}-\encircle{D} with \encircle{B}, \encircle{C}, and \encircle{D} obtained at $\zeta=0, 0.5$, and $1$, respectively), and $2)$ weighted-sum-SE joint optimization given in Algorithm~\ref{alg:One_iteration_method} (corresponding to the curve \encircle{B}-\encircle{E}-\encircle{D}). The sum-SE obtained via Algorithm~\ref{alg:One_iteration_method} upper bounds the sum-SE of the time-sharing scheme. This is because, in the former, the IRS overall phases are jointly optimal to scheduled UEs of both the MOs in any time slot. Also, the achievable sum-SE under the joint optimization peaks at the point at \encircle{$E$} when $w_1=w_2=0.5$. On the same plot, point \encircle{$A$}, which denotes the no cooperation scenario, provides a sum-SE that is smaller than that obtained by cooperation. In any case, the overall gain between the points \encircle{$A$} and \encircle{$E$} is small due to the sparse scattering of mmWave channels.
Therefore, while the presence of the OOB IRS always enhances the ergodic SE achieved by the UEs served by all MOs, the additional gain obtained via optimizing the overall phase of the IRS is marginal.  

  \vspace{-0.2cm}  
  \subsection{$M>2$-MO \& $M>2$-IRS System}
    \begin{figure}[!t]
\centering
\includegraphics[width=\linewidth]{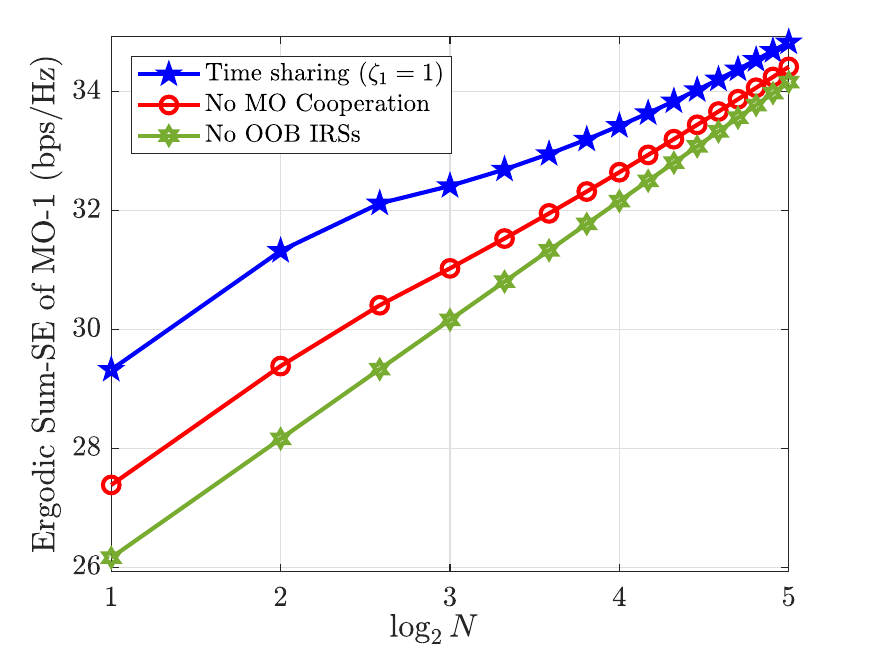}
\caption{Sum-SE of MO-$1$ vs. $\log_2 N$ with $4$-MOs, $C_0\gamma=150$ dB, $L = 5$.}
\label{fig:M_BS_M_IRS}
  \vspace{-0.2cm}
\end{figure}
We next investigate the performance obtained with more than $2$ MOs for different schemes with and without cooperation. We consider 4 MOs, with each MO deploying an IRS to serve its UEs optimally. The BSs of MO-$1,2,3,$ and $4$ are located at $(0,0)$, $(200,0)$, $(200,200)$, and $(0,200)$ (in meters), respectively, and the IRSs are located in a circular region centered at $(100,100)$ with radius $5$ meters.
The rest of the settings are the same as considered for the $2$-MO case. 

In Fig.~\ref{fig:M_BS_M_IRS}, we plot the ergodic sum-SE of MO-$1$ vs. $\log_2 N$ for $C_0\gamma = 150$ dB, and $L=5$, where $N$ is the number of IRS elements in each IRS, and investigate the SE performance of MO-$1$ for three extreme scenarios: $a)$ time-sharing with $\zeta_1=1$, $b)$ no MO cooperation scheme, and in $3$) absence of all IRSs except IRS-$1$ which is deployed by MO-$1$. We observe that the SE of MO-$1$ in the presence of OOB IRSs strictly outperforms the achievable SE in the absence of OOB IRSs. This is because, in addition to the in-band IRS, the OOB IRSs contribute to the signal strength at the UE served by MO-$1$. With cooperation, the performance can be further improved by ensuring the coherent addition of signals arriving at the UE via the contributing IRSs. 
However, for large $N$, the probability that an OOB IRS aligns with a given UE becomes small, and the SE in the presence of OOB IRSs coincides with that obtained in the absence of OOB IRSs. 
Nonetheless, even with an arbitrary number of MOs, the ergodic SE of an MO does not degrade due to the presence of uncontrollable IRSs deployed by other MOs.

\section{Conclusions}
In this paper, we addressed an important problem in IRS-aided practical mmWave wireless systems, namely, the effect of IRSs deployed by one MO on the performance of another co-existing OOB MO. Starting with the case where $2$ MOs each deploy an IRS to optimally serve its UEs, we first examined different scenarios that arise due to the impact of an IRS on the OOB MOs. Subsequently, we derived the ergodic sum-SE of the MOs under three different schemes. In the first scheme, the MOs cooperate and jointly optimize an overall phase angle of the IRSs; in the next, the MOs only cooperate by optimizing the overall phase in a time-sharing manner, and in the final scheme, the MOs do not cooperate and function independently. Our key findings were two-fold: $1$) even when the MOs do not cooperate, the IRS of one MO does not degrade the sum-SE of another MO, and $2)$ the best possible gain obtained in the sum-SE by allowing for MO cooperation compared to no cooperation scheme decreases inversely with the number of IRS elements in the OOB MO. The primary reason behind these observations is the spatial sparsity in the mmWaves band channels. This avoids degradation due to the OOB IRS and also makes significant enhancement unlikely. We extended our results to a system with more than $2$ MOs, and showed that a given MO's performance improves linearly with the number of OOB MOs in the area. 
\textcolor{black}{Future work could include extending our results to interference-limited scenarios and accounting for multi-user and inter-cell interferences, different duplexing modes, user-mobility scenarios with statistical CSI, adopting multi-user scheduling techniques~\cite{6422292, Yashvanth_TSP_2023}, etc.}
\color{black}
\bibliographystyle{IEEEtran}
\bibliography{IEEEabrv,References}

\end{document}